\documentclass{sig-alternate}
\pdfoutput=1

\usepackage{algorithm}
\usepackage[noend]{algorithmic} 
\usepackage{amsmath,amssymb}
\usepackage{pslatex}
\usepackage{amstext}

\usepackage{cite}
\usepackage{balance}
\usepackage{amsmath,amssymb}
\usepackage{theorem}
\usepackage{multirow}
\usepackage{subfigure}
\usepackage{calc}
\usepackage{amssymb}
\usepackage{amstext}
\usepackage{amsmath}
\usepackage{multicol}
\usepackage{pslatex}
\usepackage{ulem}

\newcommand{\daniel}[1]{\textcolor{red}{$<$}\scriptsize\textbf{#1}\textcolor{rée}{$>$}\normalsize}

\newcommand{\owen}[1]{\textcolor{green}{$<$}\scriptsize\textbf{#1}\textcolor{green}{$>$}\normalsize}

\newcommand{\temporary}[1]{#1}
\renewcommand{\daniel}{\cut}
\renewcommand{\owen}{\cut}

\renewcommand{\temporary}{\cut}

\newcommand{\notessential}[1]{#1}
\newcommand{\altnotessential}[1]{#1}   

\renewcommand{\notessential}{\cut}

\newcommand{\omitproof}[1]{}

\newcommand{\cut}[1]{}
\usepackage{verbatim}
\newcommand{\lxor}{\oplus}
\newtheorem{proposition}{Proposition}  

\def\qed{\relax\ifmmode\hskip2em \Box\else\unskip\nobreak\hskip1em $\Box$\fi}
\usepackage{graphicx}
\graphicspath{{.}{../BitmapIndexCpp/dbash/dolap2008experiments/}{../BitmapIndexCpp/dbash/dolap2008experiments/syntheticcolumnreordering/}{../BitmapIndexCpp/dbash/dolap2008experiments/syntheticqueries/}{../BitmapIndexCpp/synthetic/}{../BitmapIndexCpp/bash/dolap08/kingjames4d/}{../BitmapIndexCpp/bash/block_usincome/}{../BitmapIndexCpp/bash/block_dbgen/}{../BitmapIndexCpp/bash/block_netflix/}{../BitmapIndexCpp/bash/}}

\usepackage{ifpdf}
\usepackage{color}
\ifpdf
\else
\fi
\usepackage{url}

\usepackage{paralist}
\begin{document}

%
\conferenceinfo{DOLAP'08,} {October 30, 2008, Napa Valley, California, USA.}
\CopyrightYear{2008}
 \crdata{978-1-60558-250-4/08/10}

\title{Histogram-Aware Sorting for Enhanced Word-Aligned Compression in Bitmap Indexes}
%
%
%
%
%

\numberofauthors{3} 
%
\author{%
%
%
\alignauthor
Owen Kaser\\
       \affaddr{Dept. of CSAS}\\
       \affaddr{University of New Brunswick}\\
      \affaddr{100 Tucker Park Road}\\
       \affaddr{Saint John, NB, Canada }\\
       \email{o.kaser@computer.org}
\alignauthor
Daniel Lemire\\
\affaddr{LICEF,  Universit\'e du Qu\'ebec \`a Montr\'eal}\\
      \affaddr{100 Sherbrooke West}\\
        \affaddr{Montreal, QC, Canada }\\
        \email{lemire@acm.org}
\alignauthor
Kamel Aouiche\\
\affaddr{LICEF, Universit\'e du Qu\'ebec \`a Montr\'eal}\\
      \affaddr{100 Sherbrooke West}\\
        \affaddr{Montreal, QC, Canada }\\
        \email{kamel.aouiche@gmail.com}
}

\maketitle
\begin{abstract}
Bitmap indexes must be compressed to reduce input/output costs and minimize CPU usage.
To accelerate logical operations (AND, OR, XOR) over bitmaps, we use techniques
based  on run-length encoding (RLE), such
as Word-Aligned Hybrid (WAH) compression.
These techniques are sensitive to the order of the rows: a simple lexicographical
sort can divide the index size by 9 and make indexes several times faster.
We investigate reordering heuristics based on computed attribute-value histograms.
Simply permuting the columns of the table
 based on these histograms can increase 
 the  sorting efficiency by 40\%.
\end{abstract}



\cut{ for 2008, it is keywords are optional in the document
}

\section{Introduction}

Bitmap indexes are among the most commonly used indexes in data warehouses~\cite{davis2007idw,bellatreche2007sap}.
Without compression, bitmap indexes can be impractically large and slow.
Word-Aligned Hybrid (WAH)~\cite{wu2006obi} is a competitive compression technique:
compared to LZ77~\cite{chan1998bid} and Byte-Aligned Bitmap Compression (BBC)~\cite{874730},
WAH indexes can be ten times faster~\cite{502689}.

Run-length encoding (RLE) and similar encoding schemes (BBC and WAH) make it possible to compute
logical operations between bit\-maps in \cut{a}%
 time proportional to the compressed size of
the bitmaps.
However, their efficiency depends on the order of the rows.  
While computing the best ordering is NP-hard~\cite{bda08}, 
simple heuristics such as lexicographical sort are effective.

Pinar et al.~\cite{pinar05}, Sharma and Goyal~\cite{sharma2008emc}, and Canahuate et al.~\cite{pinarunpublished} \cut{ used shorter word to fix
a line break problem: applied} used Gray-code
row sorting to improve RLE and WAH compression.
However, their largest bitmap index 
 could
fit uncompressed in RAM on a PC\@. 
%
\notessential{%
We found that partially sorting the table---by blocks---was inefficient and could even
increase the indexing time.}%

We distinguish two types of heuristics for this problem. Heuristics such as
lexicographical sort~\cite{bda08} or Gray-code sorting~\cite{pinar05} are histogram-oblivious.
They ignore the number of attribute values and their frequencies.
Other heuristics are histogram-aware. They include column reorganizations and frequency-aware
ordering. 
\owen{moved to avoid ``used before defined'' on term ``histo obliv''}
On larger data sets~\cite{bda08}, we had considered  histogram-oblivious 
row-ordering heuristics. Sorting before indexing 
reduced the total construction time.
Our main contribution is an evaluation of
practical histogram-aware heu\-ris\-tics to the row ordering problem.
Secondary contributions include guidelines about when
``unary''
bitmap  encoding is preferred, and an improvement over
the naive bitmap construction algorithm---it is now practical to construct
bitmap indexes over tables with hundreds of millions of rows and millions of attribute values.

To further reduce the size of bitmap indexes, we can bin the
attribute values~\cite{354819,1155030,stockinger2004esb,1183529}. For range queries,
different bitmap encodings have different space-performance tradeoffs~\cite{chan1998bid,chan1999ebe}. 
\notessential{%
O'Neil et al.~\cite{oneil2007bid} have shown experimentally that compression and
vertical partitioning of the bitmaps are especially important when indexing read-only data.}

\temporary{%
\subsection{Notation}
This is for our sanity and should be removed prior to submission.
\begin{tabular}{|rl|} \hline
$r_i$ & ith row\\
$k$ & k (coding)\\
$N$ & num bitmaps in an encoding\\
$L_i$ & num bitmaps for encoding dimension $i$\\
$L$ & total bitmaps\\
$c$\cut{, $d$} & num columns/dimensions\\
$n$ & num rows\\
$r$ & parameter representing numbe of 1 bits in column\\
$n_i$ & distinct values for attribute $i$\\ 
$\delta$ & total number of dirty words in a column\\
$f(a_i)$ & frequency of value $a_i$\\ 
\hline
\end{tabular}
eventual consistency checks:
terminology: attribute vs dimension vs column
}


\section{Bitmap Indexes}\label{sec:bitmapIndexes}

We find bitmap indexes in several database systems, apparently beginning
with the MODEL~204 engine, commercialized for the IBM~370 in 1972.
\cut{
\owen{not sure following comment fits well; maybe we can omit it.}
One of the benefits of bitmap indexes is to favor sequential disk access~\cite{jurgens2001tbi}.
\cut{Specialized bitmap indexes~\cite{chan1999ebe} can support range queries, although the current
paper does not. }%
}%

\notessential{%
Bitmap indexes are widely applicable. 
Whereas it is commonly reported~\cite{hammer2003cea} that bitmap indexes
are suited to small dimensions such as gender or marital status,  they
also work well over large dimensions~\cite{oraclevivekbitmap,wu2006obi,bda08}. 
And as the number of dimensions increases, bitmap indexes become competitive
against  specialized multidimensional index structures such as R-trees~\cite{671192}.
}

The simplest and most common method of bitmap indexing associates a bitmap with every
attribute value $v$ of every attribute $a$;
the bitmap represents the predicate $a=v$. 
 For a table with $n$~rows (facts) and $c$~columns (attributes/dimensions), 
each bitmap has length $n$. 
Initially, all bitmap values are set to 0.
For row $j$, we set the $j^{\textrm{th}}$~component of $c$~bitmaps to 1.
 If the $i^{\textrm{th}}$~attribute has $n_i$ possible values, we have
 $L = \sum_{i=1}^{c} n_i$~bitmaps.  

Bitmap indexes
are fast, because we find rows having a given value $v$ for attribute $a$ by
reading only the bitmap
corresponding to value $v$ (and not the other bitmaps for attribute $a$), and there is
only one bit (or less, with compression) to process for each row.
 More complex queries are achieved
with logical operations (AND, OR, XOR, NOT) over bitmaps and
current microprocessor can perform 32 or 64 bitwise operations 
in a single machine instruction.

For row $j$, exactly one bitmap per column will have its  $j^{\textrm{th}}$ entry set to 1. Although the entire index 
has $nL$ bits, there are only $nc$~1's; for many tables, $L \gg c$ and thus on average
the table is very sparse.  
Long (hence compressible) 
runs of 0's are expected.

One can also reduce the number of bitmaps for large dimensions. 
Given
$L$~bitmaps, there are $L(L-1)/2$~\emph{pairs} of bitmaps.   So, instead of 
mapping an attribute value to a single bitmap, we map them to pairs of bitmaps  (see Table~\ref{tab:examples1ofk}). 
We refer to this technique 
as 2-of-$N$ encoding~\cite{wong1985btf};
 with it, we can use
far fewer bitmaps for large dimensions.  For instance, with only 2,000 bitmaps,
we can represent an attribute with 2~million distinct values.
But the average
bitmap density is much higher with 2-of-$N$ encoding, and thus compression
may be less effective.
More generally, $k$-of-$N$ encoding allows $L$~bitmaps to represent $L \choose k$
distinct values; conversely, using $L=\lceil k n_i^{1/k}\rceil$~bitmaps is sufficient to
 represent $n_i$~distinct values.
 However, searching for a specified value $v$ no longer
requires scanning a single bitmap.  Instead, the corresponding $k$~bitmaps must
be combined with a bitwise AND\@.  
There is a tradeoff between index size and the
index speed~\cite{bda08}.

For small dimensions, 
using $k$-of-$N$ encoding may fail to reduce the number
of bitmaps, but still reduce the performance. \cut{For this reason}
We apply  
the following heuristic. Any column with less than 5~distinct values is limited
to 1-of-$N$ encoding (simple or unary  bitmap). Any column with less than 21~distinct values,
is limited to $k=1,2$, and any column with less than 85~distinct values is limited
to $k=1,2,3$.

\begin{table}
\caption{\label{tab:examples1ofk}Example of 1-of-N and 2-of-N encoding}
\centering
\begin{tabular}{lcc}
Montreal  & 100000000000000  & 110000 \\
Paris     & 010000000000000  & 101000 \\
Toronto   & 001000000000000  & 100100 \\
New York  & 000100000000000  & 011000\\
Berlin    & 000010000000000  & 010100
\end{tabular}
\end{table}

\section{Compression}\label{sec:compression}
RLE 
compresses efficiently
when there are long runs of
identical values: it works by replacing any repetition by the number
of repetitions followed by the value being repeated. For example, 
the sequence  11110000 becomes 4140. Current microprocessors
perform operations over words of  32 or 64 bits and not individual bits. Hence,
the CPU cost of RLE might be large~\cite{stockinger2002spa}.
By trading some compression for more speed, Antoshenkov~\cite{874730} defined a RLE variant working over bytes
instead of bits: the Byte-Aligned Bitmap Compression (BBC).
Trading even more compression for even more speed,
Wu et al.~\cite{wu2006obi} proposed the Word-Aligned Hybrid (WAH).
Their scheme is made of two different types of words\footnote{For
simplicity, we limit our exposition to 32~bit words.}. 
The first bit of every word 
 distinguishes a verbatim (or \emph{dirty}) 31-bit word
from a running sequence of 31-bit \emph{clean} words (0x00 or 1x11). Running sequences are  stored using 1~bit to distinguish between the type
of word (0 for 0x00 and 1 for 1x11) and 30~bits to represent the number of consecutive
clean words. 
Hence, a bitmap of length 62 containing a single 1-bit at position $32$ would be coded as
the words 100x01 and 010x00.  
Because dirty words are stored in units of 31~bits using 32~bits, 
WAH compression can expand the data by 3\%.
We created our own WAH variant
called Enhanced Word-Aligned Hybrid (EWAH).
Contrary to WAH compression, EWAH may never (within 0.1\%) generate a compressed
bitmap larger than the uncompressed bitmap.
It also uses only two types of words (see Fig.~\ref{EWAHfig}). The first type is a 32-bit verbatim word. The second type of word is a
marker word: the first bit is used to indicate which clean word will follow,
16~bits to store the number of clean words, and 15~bits to store the number of dirty words following the clean words.
\notessential{%
Because EWAH uses only 16~bits to store the number of clean words, it may be less efficient
than WAH when there are many consecutive sequences of  $2^{16}$~identical clean words.}%
EWAH  bitmaps begin with a marker word.

\begin{figure}
\centering
{%
	\includegraphics[width=0.85\columnwidth]{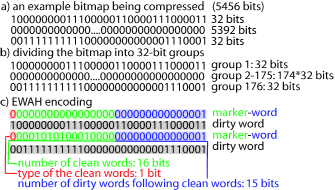}
}
\caption{\label{EWAHfig}Enhanced Word-Aligned Hybrid (EWAH)}
\end{figure}

Given $L$~bitmaps and $n$~rows, we can naively construct a bitmap index in
time $O(n L)$ by appending a word to each compressed bitmap  every 
32 or 64~rows.
We found this approach impractically slow when $L$ was large---typically, with
$k=1$.
Instead, we construct bitmap indexes in time $O(nck+L)=O(nck)$~\cite{bda08}
 where $ck$ is the number
of true values per row (See Algorithm~\ref{algo:owengenbitmap}): within each block of 32~rows, we store the values of the bitmaps
in a set---omitting any unsolicited bitmap, whose values are all false (0x00). 
We partition the table horizontally into blocks indexed with  compressed bitmaps
using 
a fixed memory
budget (256\,MiB). Each block of bitmaps is written sequentially and preceded
by an array of 4-byte integers containing the location of each bitmap.

\begin{algorithm}[tb]
\small
\begin{algorithmic}
\STATE Construct: $B_1,\ldots, B_L$, $L$ compressed bitmaps 
\STATE $\omega_i \leftarrow 0$ for $1 \leq i \leq L$.
\STATE $c\leftarrow 1$ \COMMENT{row counter}
\STATE $\mathcal{N} \leftarrow \emptyset$ \COMMENT{$\mathcal{N}$ records the dirtied bitmaps}
\FOR{each table row}
\FOR{each attribute in the row}
\FOR {each bitmap  $i$ corresponding to the attribute value}
\STATE set to true the $(c \bmod w)^{\textrm{th}}$~bit of word $\omega_i$
\STATE $\mathcal{N} \leftarrow \mathcal{N} \cup \{i\}$
\ENDFOR
\ENDFOR
\IF{$c$ is a multiple of $w$}
\FOR{$i$ in $\mathcal{N}$}
\STATE add $c/w-\vert B_i \vert  - 1$~clean words  (0x00) to $B_i$
\STATE add the word $\omega_i$ to bitmap $B_i$
\STATE $\omega_i \leftarrow 0$
\ENDFOR
\STATE $\mathcal{N} \leftarrow \emptyset$
\ENDIF
\STATE  $c\leftarrow c+1$
\ENDFOR
\FOR{$i$ in \{1,2,\ldots,L\}}
\STATE add $c/w-\vert B_i \vert  - 1$~clean words (0x00) to $B_i$ 
\ENDFOR
\end{algorithmic}
\caption{\label{algo:owengenbitmap}
Constructing bitmaps. For simplicity, we assume the number of
rows is multiple of the word size.}
\end{algorithm}

Naively, we could compute logical operations between 2~bitmaps in n/32~bitwise
operations. Instead,  we compute
logical operations (OR, AND, XOR) between 2~bitmaps in time
$O( \vert B_1 \vert+\vert B_2 \vert)$ where $\vert B_i \vert$ is the size
of the compressed bitmap~\cite{wu2006obi,bda08}. Finally, we can bound the bitmap sizes:
$\vert \bigwedge_i B_i  \vert \leq \min_i \vert B_i\vert$
and  
$\vert \bigvee_i B_i \vert \leq \sum_i\vert B_i\vert$.

\section{Sorting to Improve Compression}

Sorting can benefit bitmap indexes at several levels.
We can sort the rows of the table.
The sorting order depends itself on the order of the table columns.
And finally, we can allocate the bitmaps to the attribute values in sorted order.

\subsection{Sorting rows}
\label{sec:sorting-rows}

Reordering the rows of a compressed bitmap index can improve
compression. Whether \cut{you use}%
 using RLE, BBC, WAH or EWAH,
the problem is NP-hard by reduction from the Hamiltonian path
problem~\cite[Theorems 1 and 2]{bda08}.
A simple heuristic begins with an uncompressed
index.  Rows (binary vectors)
are then rearranged to promote runs.  In the process, we may also
reorder the bitmaps.
%
This is the approach of Canahuate et al.~\cite{pinarunpublished},
but it uses $\Omega(n L)$ time.
For the large dimensions and number of rows we have considered,
it is infeasible. 
A more practical approach~\cite{bda08} is to reorder the  table,
then construct the compressed index directly; we can also reorder the 
table columns prior to sorting.

Three types of ordering can be used for ordering rows.
We may cluster identical rows, but it is not a competitive heuristic~\cite{bda08}.
\begin{itemize}
\item  In lexicographic order, a sequence
$a_1, a_2, \ldots$ is smaller than another sequence $b_1, b_2, \ldots$ if
and only if there is a $j$ such that $a_j < b_j$ and $a_i = b_i$ for $i<j$.
The Unix \texttt{sort} command provides an efficient
mean of sorting flat files into lexicographic order; in under 10\,s
our test computer (see Section~\ref{sec:Experiment})  sorted a
5-million-line,
120\,MB file. SQL supports lexicographic sort 
via ORDER BY. 
\cut{\item We may cluster runs of identical  rows. This problem can be solved
with hashing algorithms, by multiset discrimination algorithms~\cite{cai1995umd}, or by a
lexicographic sort. It is not a competitive heuristic~\cite{bda08}.}
\item Gray-code (GC) sorting is defined over bit vectors~\cite{pinar05}:
the sequence $a_1,a_2, \ldots$ is smaller than  
$b_1,b_2,\ldots $ if and only if there exists $j$ such that\footnote{The symbol $\lxor$ is the XOR operator.} 
$a_j = a_1 \lxor a_2 \lxor \ldots \lxor a_{j-1}$,
$b_j \not = a_j$, and $a_i=b_i$ for $i<j$. 
Algorithm~\ref{algo:graycomp} shows how to compare sparse GC bit vectors
$v_1$ and $v_2$ in time $O(\min(\vert v_1 \vert,\vert v_2 \vert)$ where
$ \vert v_i\vert$ is the number of true value in bit vector $v_i$.
Sorting the rows of a bitmap index without materializing the uncompressed
bitmap index is possible~\cite{bda08}: we 
implemented an $O(n c k \log n)$-time solution
for $k$-of-$N$ indexes using
an external-memory B-tree~\cite{qdbm}. \cut{As values, we used the rows of the table, and as keys, we used
the position of the ones in the bitmap row. Both were compressed using LZ77 to minimize IO costs.}%
 Unfortunately, it proved to be two orders of
magnitude slower than lexicographic sort.
\end{itemize}

\begin{algorithm}
\begin{algorithmic}
\STATE \textbf{INPUT}: arrays $a$ and $b$ representing the position of the ones in two bit vectors
\STATE \textbf{OUTPUT}: whether the bit vector represented by $a$ is less than the one represented by $b$
\STATE $f\leftarrow \texttt{true}$
\STATE $m \leftarrow \min(\textrm{length}(a), \textrm{length}(b))$
\FOR{$p$ in $1,2,\ldots,m$}
\STATE return $ f$ if $a_p > b_p$ and $ \lnot f$ if $a_p < b_p$
 \STATE $f\leftarrow \lnot f$
 \ENDFOR
 \STATE return $\lnot f$ if $\textrm{length}(a) > \textrm{length}(b)$,
 $f$ if $\textrm{length}(b)>\textrm{length}(a)$, and  \texttt{false} otherwise
\end{algorithmic}
\caption{\label{algo:graycomp}Gray-code less comparator between  sparse bit vectors}
\end{algorithm}

For RLE, the best ordering of the rows
of a bitmap index 
minimizes the sum of the Hamming distances:
$\sum_i h(r_i,r_{i+1})$ where $r_i$ is the $i^{\textrm{th}}$~row, for
$h(x,y) = |\{i|x_i\not= y_i\}|$.
If all $2^L$~different rows are present, the GC sort would
be an optimal solution to this problem~\cite{pinar05}. The following
proposition shows that GC sort is also optimal
if all ${N \choose k}$ $k$-of-$N$ codes are present. The same is 
not true of lexicographic order when $k>1$: 0110 immediately follows 1001 among
2-of-4 codes, but their Hamming distance is 4.

\begin{proposition}\label{prop:graycode}We can enumerate,  in GC order, all $k$-of-$N$ codes
in time $O(k {N \choose k})$ (optimal complexity).
Moreover, the Hamming distance between successive codes is minimal (=2).
\end{proposition}
\begin{proof}
Let $a$ be an array of size $k$ indicating the positions of the ones in $k$-of-$N$ codes.
As the external loop, vary the value $a_1$ from 1 to $N-k+1$. Within this loop,
vary the value $a_2$ from $N-k+2$ down to $a_1+1$. Inside this second loop,
vary the value of $a_3$ from $a_2+1$ up to $N-k+3$, and so on.
By inspection, we see that all possible  codes are generated in increasing GC order.
%
To see that the Hamming distance between successive codes is 2, consider what
happens when $a_i$ completes a loop. Suppose that $i$ is odd and greater than 1,
 then $a_i$ 
had value $N-k+i$ and it will take value $a_{i-1}+1$. Meanwhile, by construction,
 $a_{i+1}$ (if
it exists) remains at value $N-k+i+1$ whereas $a_{i+2}$ remains at value
$N-k+i+2$ and so on. The argument is similar if $i$ is even.
\end{proof}

\begin{figure}
\centering
\includegraphics[width=0.8\columnwidth]{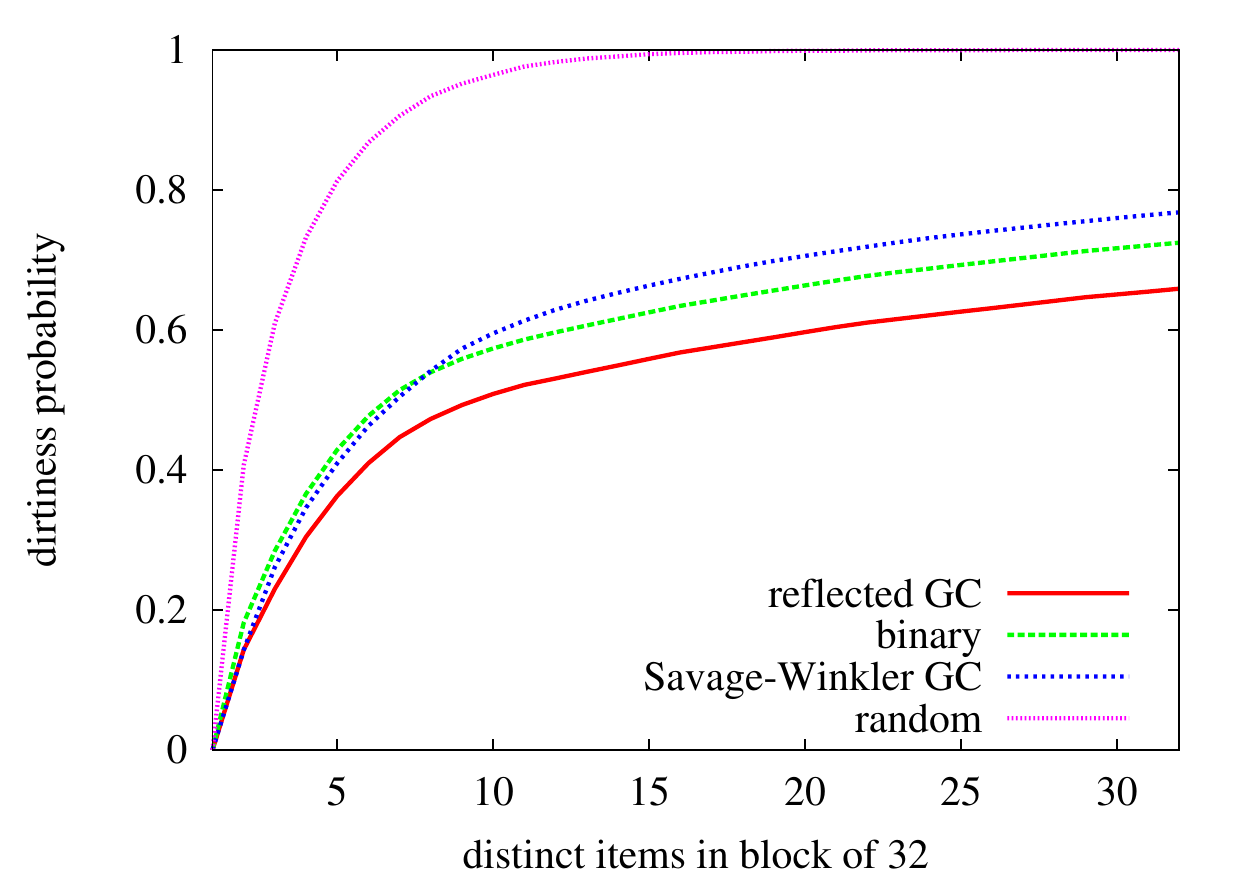}
\caption{\label{fig:adjdirty-thousand}Probabilities that a bitmap will contain a dirty word, when
several ($x$-axis) of 1000 possible 
attribute values are found in a 32-row chunk.
Effects are shown for values with $k$-of-$N$ codes that
are adjacent in GC order, adjacent
in lexicographic order, or randomly selected. 
}
\end{figure}

For a given column, suppose that in a block of 32~rows, we have $j$~distinct attribute values.
We computed the average number of bitmaps
that would have a dirty word  (see Fig.~\ref{fig:adjdirty-thousand}).
Comparing $k$-of-$N$ codes that were adjacent in GC ordering against
$k$-of-$N$ codes that were lexicographically adjacent, the difference was
insignificant for $k=2$.   However, GC ordering is substantially better
for $k>2$, where bitmaps are denser.  Selecting the codes randomly is
disastrous. 
Hence, sorting part of a column---even one without long runs of identical
values---improves compression for $k>1$.

For encodings like BBC, WAH or EWAH, GC sorting is not optimal,
even when all $k$-of-$N$ codes are present. For example consider the
sequence of rows 1001, 1010, 1100, 0101, 0101, 0110, 0110, 0011. Using 4-bit
words, we see that a single bitmap contains a clean word (0000)
whereas by exchanging the fifth and second row, we get two clean words (0000
and 1111).

\subsection{Sorting bitmap codes}
For a simple 
index, the map from  attribute value to bitmaps is
inconsequential; for $k$-of-$N$ encodings, some bitmap allocations are more
compressible: consider an attribute with
two overwhelmingly frequent values and many other values that
occur once each.  If the table rows are given in random order,
the two frequent values should have codes that
differ 
as little 
as possible.

There are several ways to allocate the bitmaps.
Firstly, the attribute values can be visited in alphabetical or numerical order,
or---for histogram-aware schemes---in order of frequency. Secondly, the bitmap codes
can be used in different orders. We consider lexicographical ordering
 (1100, 1010, 1001, 0110, \ldots) and GC
order (1001, 1010, 1100, 0101, \ldots) ordering (see proof of Proposition~\ref{prop:graycode}). 
For dense low-dimensional tables, GC
order is preferable~\cite{bda08} 
 and its compression effects
are comparable  to sorting
the index \emph{rows} in GC order. 
Meanwhile, it is technically easier to implement
since we can sort the table lexicographically and only use GC ordering
during the bitmap index construction.

\emph{Alpha-Lex} denotes sorting the table lexicographically 
and assigning bitmap codes so that
the $i^{\textrm{th}}$ attribute gets the lexicographically $i^{\textrm{th}}$
smallest bitmap code.  \emph{Gray-Lex} is similar, except that 
the $i^{\textrm{th}}$ attribute gets the rank-$i$ bitmap code
in GC order.  These two approaches are histogram oblivious---they
ignore
the frequencies of attribute values.

Knowing the frequency 
of each attribute value 
can improve code assignment when $k>1$. 
For instance,
clustering dirty words increases the compressibility.
Within a column, {Alpha-Lex} and {Gray-Lex} order runs of identical values
irrespective of the frequency: the sequence 
\texttt{afcccadeaceabe} may  become 
\texttt{aaaabccccdeeef}. For better compression, we should order
the attribute values---within a column---by their frequency 
(e.g., \texttt{aaaacccceeebdf}). Allocating the bitmap codes in
GC order to the frequency-sorted attribute values, our \emph{Gray-Frequency}
sorts the table rows as follows.
Let $f(a_i)$ be
the frequency  of attribute $a_i$. Instead
of sorting the table rows $a_1, a_2,\ldots, a_d$, we
lexicographically sort the extended rows
$f(a_1), a_1, f(a_2), a_2,\ldots, f(a_d), a_d$
by
comparing the frequencies by their numerical value. 
The frequencies $f(a_i)$ are 
discarded prior to indexing.

\subsection{Choosing the column order}\label{sec:dimOrder}

%
%
\cut{(To be added later in the paper: we should distinguish the work of 
Canahuate et al.~\cite{pinarunpublished} where they reorder
bitmaps (which they call "columns") and our work were
we reorder "table columns".
This is important because they state: "While the overall compression ratio is not significantly affected by the order in which the 
columns are evaluated, the maximum column size, which translates to the worst case query
execution time, is reduced considerably."
In other words, they report no reduction of the index size under "column reordering". 
Our results differ substantially.
Their column reordering is akin to my bitmap interleaving.
)}%
 
Lexicographic table sorting
uses the $i^\textrm{th}$ column as the $i^{\textrm{th}}$ sort key: it
uses the first column as the main key, the second column to break ties
when two rows have the same first component, and so on.
Some column orderings
lead to smaller indexes than others~\cite{bda08}. 

We model
the \emph{storage cost}  of a bitmap index
as
the sum of the number of dirty words and 
the number of sequences of identical clean words (1x11 or 0x00).
If a set of $L$~bitmaps has $x$~dirty words, then there are at most
$L+x$~sequences of clean words; the storage cost is at most $2x+L$.
This bound will be tighter for sparser bitmaps.
Because the simple index of a column has at most $n$~1-bits, it has
at most $n$~dirty words, and thus, the storage cost is at most 
$3n$. The next proposition shows that the storage cost 
of a sorted column is bounded by $5n_i$.

\begin{proposition}\label{prop:sortingbenef}
Using GC-sorted $k$-of-$L$ codes, a sorted column with $n_i$ distinct values has no more than $2n_i$~dirty words,
and the storage cost  is no more than $4n_i+\lceil k n_i^{1/k}\rceil$.
\end{proposition}
\omitproof{%
\begin{proof}
Using $\lceil k n_i^{1/k}\rceil$~bitmaps is sufficient to represent $n_i$~values.
Because the column is sorted, we know that the Hamming distance of the bitmap
rows corresponding to two successive and different attribute values is 2. Thus
every transition creates at most two dirty words. There are
$n_i$ transitions, and thus at most $2n_i$ dirty words. This proves the result.
\end{proof}}

For $k=1$, Proposition~\ref{prop:sortingbenef} is true irrespective of the order
of the values, as long as identical values appear sequentially.
Another extreme is to assume that all 1-bits are randomly 
distributed. Then \cut{for } sparse bitmap indexes\cut{, we} %
 have
$\approx \delta(r,L,n)=(1-(1-\frac{r}{Ln})^w)\frac{Ln}{w}$~dirty words where $r$ is the number of 1-bits, $L$ is
the number of bitmaps and $w$ is the word length ($w=32$).
Hence, we have an approximate storage cost of 
$2\delta+\lceil k n_i^{1/k}\rceil$. 
The \emph{gain} of column $\mathcal{C}$ is the
difference between the expected storage cost of a randomly row-shuffled $\mathcal{C}$,
minus the storage cost of a sorted $\mathcal{C}$.
We estimate the gain  by
$2\delta(kn,\lceil k n_i^{1/k}\rceil,n)- 4n_i$ 
(see Fig.~\ref{fig:theorygain}) for columns with uniform histograms.
The gain is modal: it increases until a maximum is reached and then it
decreases.
The maximum gain is reached at $\approx \left( n(w-1)/2 \right)^{k/(k+1)}$:
%
for $n=100,000$ and $w=32$, the maximum is reached at $\approx 1,200$ for $k=1$ and
at $\approx 13,400$ for $k=2$.
Skewed histograms have a lesser gain for a fixed cardinality $n_i$. 
\cut{We can generalize our measure of gain if the per-bitmap histogram ($f_j$) 
 is
 available: set
$\delta = \sum_{j=1}^{\lceil k n_i^{1/k}\rceil}  (1-(1-\frac{f_j}{n})^w) \frac{n}{w}$. }
\cut{%
\daniel{Next bit seems intuitively clear:}
Intuitively, simple indexes ($k=1$) have
their gain diminish faster as a function of the histogram skew.
}%

\begin{figure}\centering
\includegraphics[width=0.8\columnwidth]{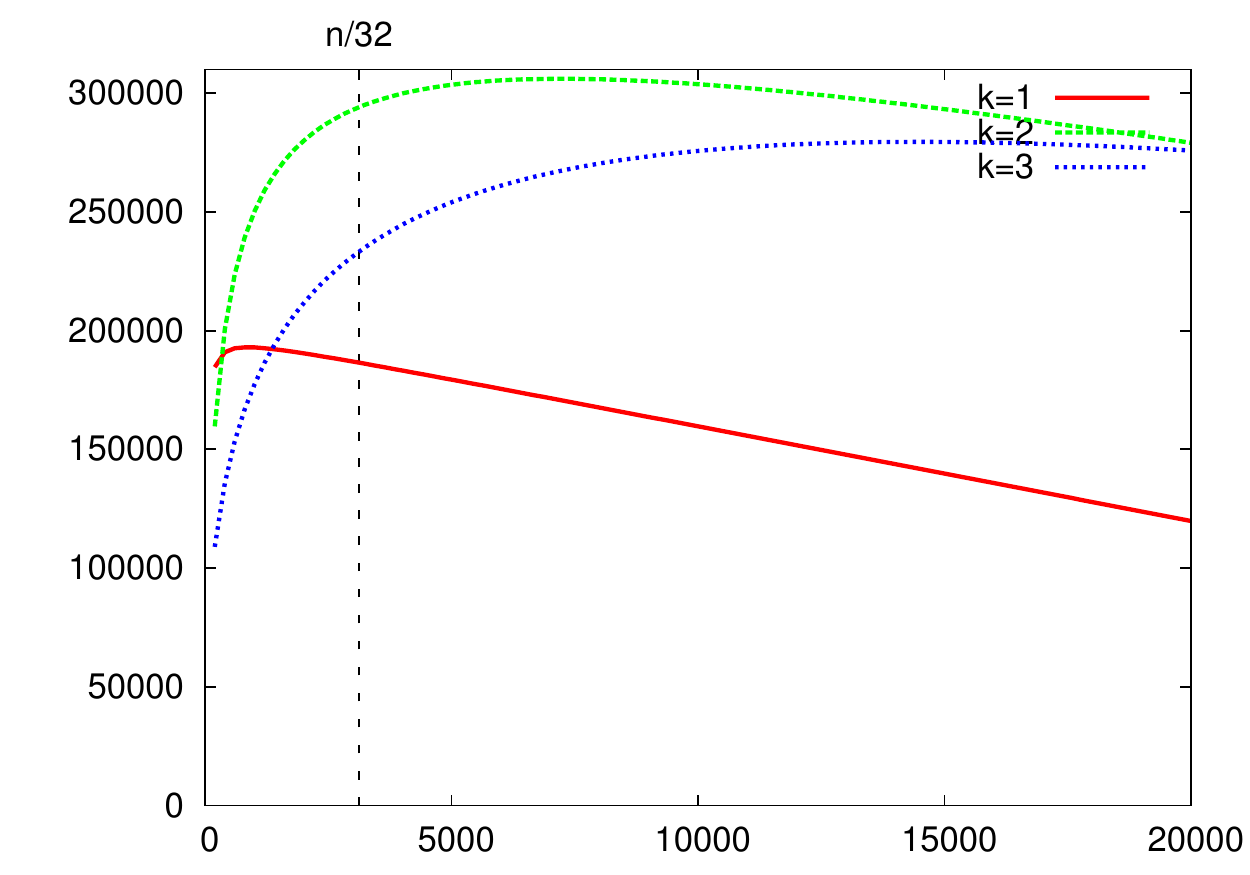}
\caption{\label{fig:theorygain}Storage gain 
in words for sorting a given column with $100,000$~rows and various number of attribute values
($2\delta(kn,\lceil k n_i^{1/k}\rceil,n)- 4n_i$ ).}
\end{figure}

After lexicographic sorting, the $i^{\textrm{th}}$~column 
is divided into at most $n_1 n_2\cdots n_{i-1}$~sorted 
blocks.
 Hence, it has at most $2 n_1 \cdots n_i$~dirty
words\cut{\ and a storage cost of at most $\min(4 n_1\cdots n_i,2\delta(kn,\lceil k n_i^{1/k}\rceil,n))+\lceil k n_i^{1/k}\rceil$~words}.
When the distributions are skewed, the $i^{\textrm{th}}$~column will have
blocks of different 
lengths and their ordering depends
 on how the columns are ordered.  \cut{For example, if the first dimension
is skewed and the second uniform, the short blocks will be clustered,
whereas the reverse is true if you exchange the columns. Clustering
the short blocks, and thus the dirty words, increases compressibility. 
Thus, it may be preferable to put skewed columns in the first positions
even though they have lesser sorting gain. }%
To assess these effects, we generated data with 4~independent  columns:  using
uniformly distributed dimensions of different sizes (see Fig.~\ref{fig:uni-order}) and 
using same-size dimensions of different skew (see Fig.~\ref{fig:zipfwide-order}). 
We then determined
the Gray-Lex 
index size for each of the 4! different dimension orderings.
Based on these results, for sparse indexes ($k=1$), dimensions should be ordered from least  to most skewed, and
from smallest to largest;
whereas the opposite is true for $k>1$. 

\begin{figure}
\centering
\subfigure[Uniform histograms with cardinalities 200, 400, 600, 800\label{fig:uni-order}]{%
\includegraphics[width=0.8\columnwidth]{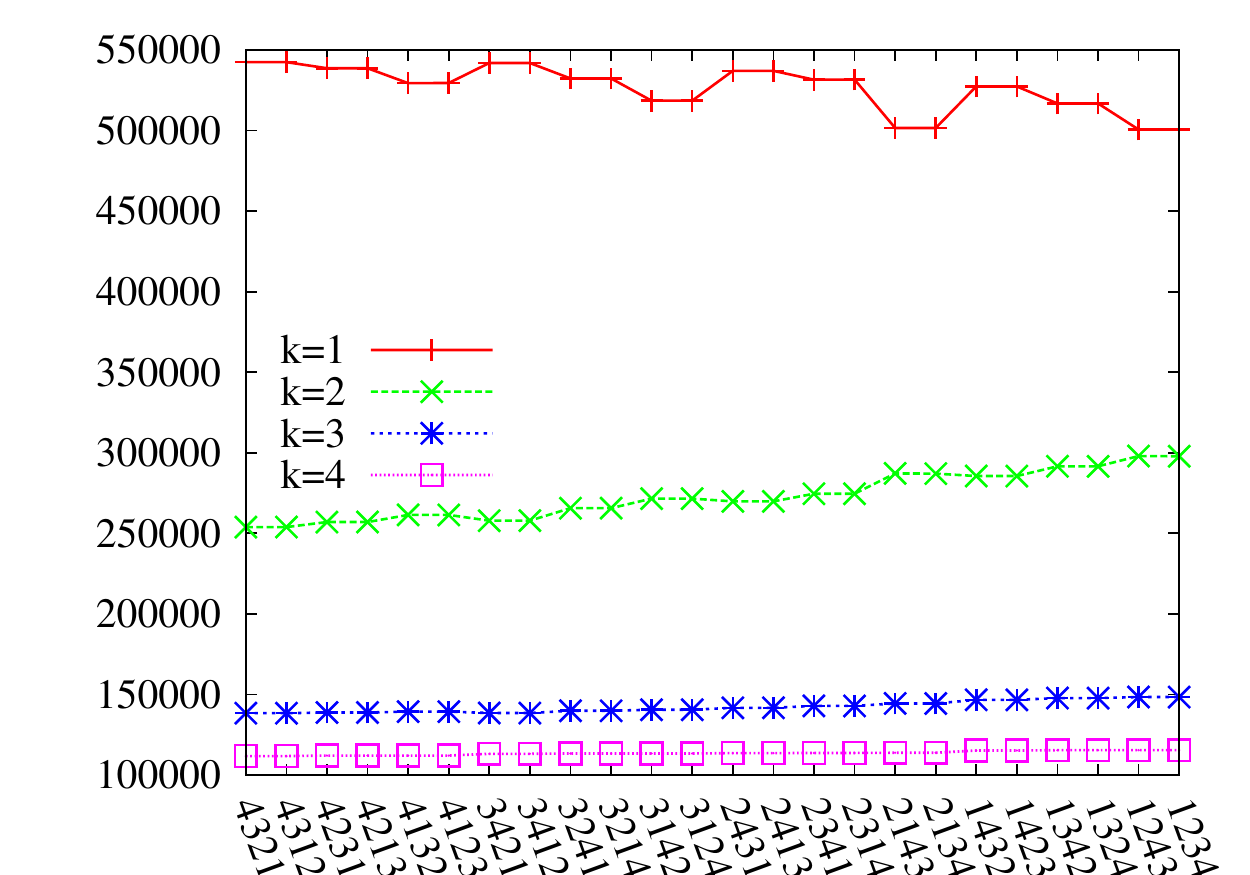}
}
\subfigure[Zipfian data with skew parameters 1.6, 1.2, 0.8 and 0.4\label{fig:zipfwide-order}]{%
\includegraphics[width=0.8\columnwidth]{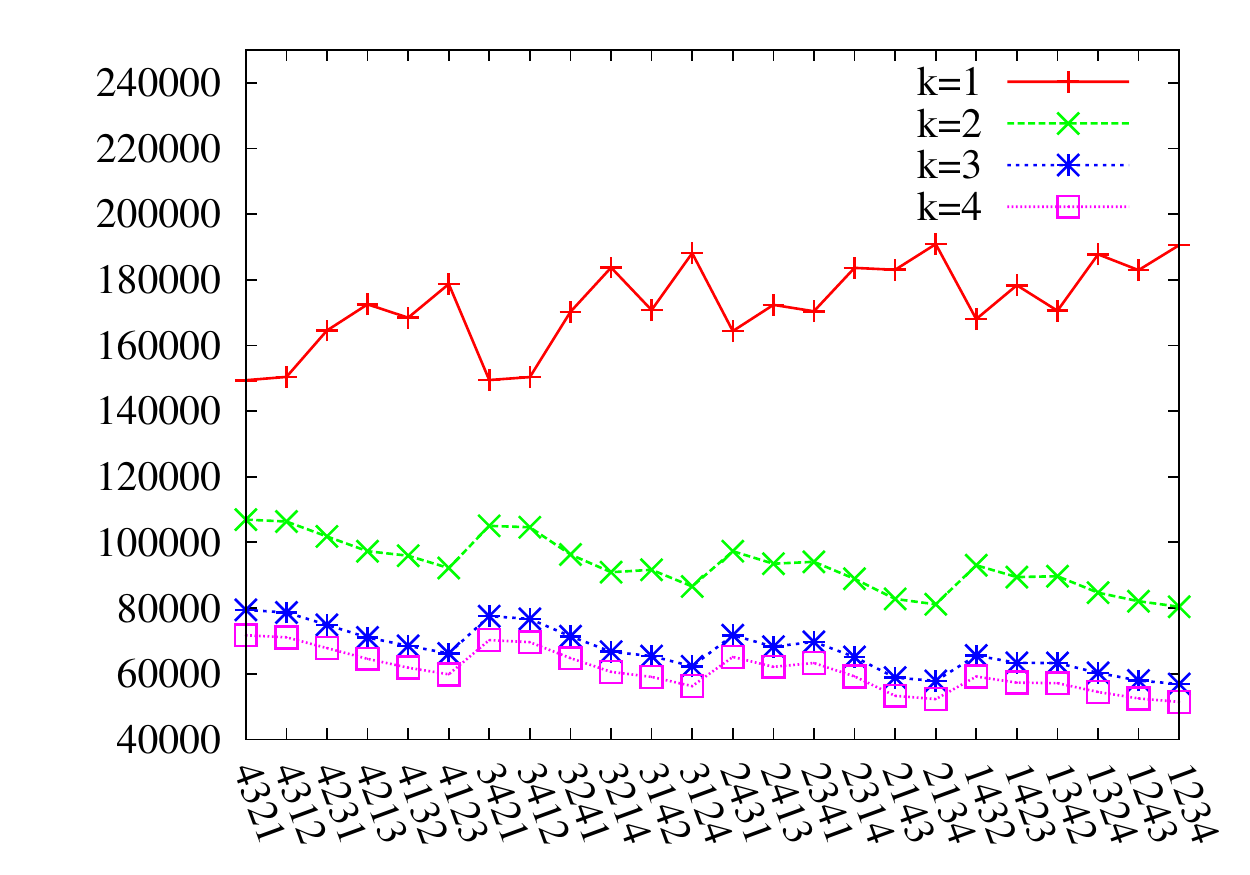}
}
\caption{Index sizes in words for various dimension orders on synthetic
 data ($100,000$ rows). Zipfian columns have 100 distinct values.\label{fig:toyorderings}
Ordering ``1234'' indicates ordering by descending skew (Zipfian) or
ascending cardinality (uniform).
%
}
\end{figure}


A sensible heuristic might be to sort columns by increasing density ($\approx n_i^{-1/k}$). 
However, a very sparse 
column  ($n_i^{1/k}\gg w$) 
will not benefit from sorting 
(see Fig.~\ref{fig:theorygain})
and should be put last. Hence, we use the following heuristic: columns are sorted in 
 decreasing order with respect to $\min(n_i^{-1/k}, (1-n_i^{-1/k})/(4w-1))$: this function 
 is maximum  at density $n_i^{-1/k}= 1/(4w)$ and it goes down to zero as the density goes to 1. 
 In Fig.~\ref{fig:uni-order}, this heuristic makes the best choice for all values of $k$.
We  consider this heuristic further
in Section~\ref{sec:exper-ordering}.
\cut{%
If we only sort using the first column as a sorting key and have
 statistically independent columns,
a column ordering with maximal first-column gain has minimal expected storage cost.
But when using lexicographic sorting, 
we might want to choose a column with little sorting gain as
the first column.
If the first column has attribute values with frequencies $f_1^{(1)}, f_2^{(1)},\ldots$, then
the second column has at most $\sum_{z=1}^{n_1} \min(k f_z^{(1)}, 2 n_2)$~dirty words assuming GC-sorted
$k$-of-$L$ codes are used.}
\cut{---for the $j^{\textrm{th}}$~column, we should consider
the frequencies of the $j-1$~tuples made by the entries of the $j-1$ first columns.
To see why this result must hold, consider a sequence of $f_z^{(1)}$~rows
for which the value of the first column is constant. Because the
values of the second column are sorted and there are at most 
$n_2$ of them, we have at most $2n_2$ dirty words by Proposition~\ref{prop:sortingbenef}.
Also, there are at most $k f_z^{(1)}$ 1-bits, and thus no more than 
$k f_z^{(1)}$~dirty words.} 
\cut{\daniel{This whole thing is awfully vague:}
This analysis suggests the following
remarks:
\begin{itemize}
\item If the histogram sizes ($n_i$) are small compared to
the frequencies of the individual attribute values ($n_i \leq k f_z^{(i-1)}$), then we can
bound the storage cost for $k=1$ of two columns by 
$4n_1+4 n_1 n_2 +n_2+\lceil k n_1^{1/k}\rceil+\lceil k n_2^{1/k}\rceil$ and it 
is maybe best to put the column with the largest histogram first.
\cut{%
$4n_1+\lceil k n_1^{1/k}\rceil+
2 \sum_{z=1}^{n_1} \min(k f_z^{(1)}, 2 n_2) +\lceil k n_2^{1/k}\rceil=
4n_1+4 n_1 n_2 +n_2+\lceil k n_1^{1/k}\rceil+\lceil k n_2^{1/k}\rceil$. Hence, it is best to put the column with the largest
histogram first. This is consistent with maximizing the gain of the first column.}
\item If  columns have relatively infrequent items ($f_z^{(i-1)}\leq n_i/k $),
then we expect that only the first column will benefit from the sorting, 
and it is best to maximize its gain.
\item However, when the first dimension has infrequent attribute
values ($f_z^{(i-1)}\leq n_i/k $) whereas the second dimension
has frequent attribute values ($f_z^{(i)}\leq n_{i-1}/k $), it might be
best to put the second dimension first. 
\end{itemize}
\daniel{End of evil vagueness?}}
\cut{%
\begin{figure}\centering
\includegraphics[width=0.8\columnwidth]{log2dskewedforowenk2}
\daniel{\\For $k=1$, it does not work:\\
\includegraphics[width=0.8\columnwidth]{log2dskewedforowen}
}\\
\daniel{A similar experiment without skew, but with different cardinalities does not work for any $k$:\\
\includegraphics[width=0.8\columnwidth]{log2duniforowenk2}
}
\caption{\label{fig:skewedforowen}Storage gain 
in words for sorting skewed tables with $n=10000$ and different skews ($k=2$)}
\end{figure}
}
\cut{%
\subsubsection{Uniform histograms}
\label{subsubsectionunihisto}
\daniel{Daniel has the conjecture that for uniform histograms, 
sorting by the gain is good enough. This is \textbf{proven} by the
following figure:\\
\includegraphics[width=0.8\columnwidth]{uniformhistograms}\\
In other words, if you can assume that the histograms are
flat, just sort the columns in inverse order of the gain and
you will be fine. We should be able to prove this analytically.}
}

\cut{%
\subsubsection{Simple bitmap indexes}
\daniel{Daniel had the conjecture that for $k=1$, basically, you do not need
to worry about the effect of the bias. I proved this false with
the simpleindex.sh script. So we can drop this subsubsection.}
}

\cut{%
\subsubsection{Measuring the bias of the histogram}
\daniel{A natural bias measure is the difference between the
gain measured above from $k$ and $n_i$, with a uniform histogram assumption,
and the (lower) gain measured by taking into account the full histogram.}
\daniel{Indeed, using the result from the subsubsection~\ref{subsubsectionunihisto},
we have that when the bias is zero, according to this measure,
then you should just sort by gain because the distributions are uniform.}
\daniel{A wild conjecture is that whatever a column loses in gain
due to its bias, it makes it up by allowing the following columns
to be better sorted.}
}

\cut{\owen{Probably obsolete discussion: Looking at experimental(skewed), it appears that the conclusion
for $k=1$ is unaffected: n/32 gets best gain.  For k=2 uniform, and k=3
(skewed or uniform),  your theoretical curves suggest bigger is better.
But for k=2 skewed, there appears to be a gentle peak around 1500.\\
Daniel's figures 3b can be compared to Owen's~\ref{fig:zipfs2-order}
and 3c to~\ref{fig:zipfs10}.  For k=1, Owen's conclusion can be that
the same orders are suitable but the difference is less extreme
when skew is high.  Daniel's does not seem to contradict this: gain
has similar peak location, but is less extreme.  For both Daniel and
Owen, ramping up the skew does not change the conclusion for $k>1$
that big dimensions should go first.  For Owen, k=2, the difference
between the best and worst ordering was a factor of 2, regardless of
large or small skew.  Daniel's result would have lead to the guess that
ordering was less important for higher skew.  Hmm.}
}

\cut{\owen{Experimentally, skew seems to matter a lot (see pix toward end), but
 affects $k=1$ in an opposite way from how it affects $k>1$.
 For sanity, maybe we are stuck with a distribution oblivious model here,
 but somehow, eventually, we ought to try to knit the skew effects 
(whether we can  explain it mathematically or just describe what we see)
together with these results relating table length, $k$ and dimension
size.}}
\cut{\daniel{I have added skewed distributions in my figures. It seems that column
order is less important if the distributions are very skewed. Which makes sense
if you think about it. Imagine a distribution where all values occur once, except
for one that is repeated often. If the number of distinct attribute values,
this is nearly as having just one attribute value, and then, clearly, column order
does not matter.}}

%

\cut{\daniel{An obvious question maybe is why not sort by gain as a heuristic? This would tend
to optimize mostly the first few dimensions. But looking at table 7 from BDA,
it looks like sorting by gain might do the right thing. 
\owen{After sorting by some first dimension, you essentially create a
bunch of sort-able sub problems, one for each attribute value in the 
first dimension.  To choose the second sort dimension, you should not
apply the original gain formula.  Instead, the number of rows should be
adjusted to be suitable for the sizes of the sub-problems.  We could 
choose average, which is appropriate if the first dimension was uniform.
However, if the first dimension was fairly big and uniform, it does not
much matter how you order the remaining dimensions.  So try to help
skewed ones.  I suggest that you let $n$ be the number of occurrences
of the most frequent value in dimension 1, for choosing dimension 2.
And for choosing dim 3, let $n$ be the number of occurrences of the 
most frequent pair
$(v1,v2)$ where $v1$ is from the first dimension, $v2$ is from the
second dimension. Etc.\\
If you just generate gains for all dims using the same $n$, then 
the choice made for the subproblems will not avoid excessively large
dimensions.
}
The histogram-aware
formula is really not that hard:
$\sum_{j=1}^{\lceil k n_i^{1/k}\rceil} 2 (1-(1-f_j/n)^w) n/w+\lceil k n_i^{1/k}\rceil-
2n_i$ (assuming I did not screw this up) where $f_j$ represents the per-bitmap
histogram.
One can validate that the formula is roughly sensible by taking each
column, shuffling it, computing the size, the subtracting the 
size of the index over the sorted column.
}
}

\cut{%
In Section~\ref{sec:exper-ordering}, we evaluate the following heuristics:
\daniel{For each heuristic, we could plot a corresponding cost measure that it
seeks to minimize, something like $n^{d-1} \min(n_1, n/32)+ n^{d-2} \min(n_2, n/32)+\ldots$.  }
\begin{enumerate}
\item sort the columns in decreasing order with respect to $\min(n_i,k^2 n (w-1) / (2\max_{j=1,\ldots,d} n_j)$;
\item sort the clumsy in decreasing order with respect to $\min(n_i, n/32)$ \daniel{If we are going
to set an arbitrary point such as $n/32$, why not pick $\approx \left( \frac{w-1}{w} \right)^{\frac{k}{k+1}} n^{k/(k+1)}$?};
\item sort the clumsy in decreasing order with respect to the number of attribute values
occurring more than 32~times.
\end{enumerate} 
\cut{As a heuristic, in Section~\ref{sec:exper-ordering} we evaluate the following approach, which 
requires knowing only the schema information and the number of rows, $n$:
choose as the primary sort key the dimension $D_i$ with largest
$n_i$, as long as $n / n_i \geq 32$.  
\owen{vaguely resembles the histo approach sketched at the end of 4.3 of BDA, I guess}
The second sort key is the 
next largest dimension, and so forth.  This
is our \textbf{ordering heuristic 2}.  Another and similar 
heuristic~\cite{bda08} requires
computing histograms: we favour the dimensions in order of
descending size, where we consider only values that have more than 32 occurrences. 
 This is \textbf{ordering heuristic 3}.}
}
\cut{%
If $k=1$, favoured dimensions should
have low skew (calculate with statistical methods),
except that it should not begin with a small dimension followed
by a dimension that does not have some highly frequent
members.    

If $k>1$, favoured dimensions should have high skew and large size.
\owen{Some combination is required?}
As $k$ increases, size becomes less significant.
\owen{Some combination of skew and the kth root of the dim size,
representing the number of bitmaps?}

This leads to \textbf{ordering heuristic 4:}  For $k=1$, order
dimensions by increasing skew (need a quick test for that).
For $k>1$, for each dimension $D$ we
compute $g_D |D|^{1/k}$, where $g_D$ is $<$some statistical test that
estimates skew in a reasonable way$>$.  Order dimension by
decreasing value of this product.
}

\subsection{Avoiding column order}  

\cut{%
%
Although we suggest first sorting the  table and then encoding
the sorted table into a compressed bitmap index, our approach appears
less flexible than first materializing an uncompressed index,
then sorting, then compressing. 
However, the B-tree method sketched in Section~\ref{sec:sorting-rows}
yields an equivalent result, enabling a more equal treatment of
(primary sort key, second sort key, and so forth).  
\cut{Except for the first
few dimensions, many dimensions do not have a large effect on 
the sort.   We can treat dimensions more equally if we consider the
individual bitmaps arising from the dimensions.}
Suppose dimension $i$ leads to bitmaps $b_{i,1}, b_{i,2}, \ldots b_{i,L_i}$.
One ``interleaving'' scheme for would be to sort the index lexicographically
by $b_{1,1}, b_{2,1}, \ldots, b_{D,1}, b_{1,2}, b_{2,2}, \ldots$.
(Canahuate et al.\  similarly propose
to order individual bitmaps by their densities~\cite{pinarunpublished}.)

\notessential{%
It \emph{is} possible to sort the  table in a bitmap-interleaved
fashion, without materializing the bitmap index first.   Any
general-purpose external-memory sorting algorithm may be used, providing
we may specify how rows are compared.   To compare rows $r_1$
and $r_2$, we temporarily convert the rows to their corresponding
lines in the bitmap index, and then compare the two lines as desired
(by Gray code ordering, or lexicographically according to the densest
bitmaps, etc.)  Although we know that O($n \log n$) comparisons will be
sufficient, each comparison can be expensive:  it entails encoding
into a long (probably sparse) vector.  \owen{we may have
already hit this enough earlier} A similar idea is to store
each row and (using sparse vector methods) its bitmap line, in an
external memory data structure for ordered sets: we used the B-trees 
provided by QDBM~\cite{qdbm}.
\owen{we could translate the stuff at end of BDA section 3 to elaborate.}
The key-comparison routine is then based on the stored bitmap line,
which only needs to be calculated once.  The disadvantage is that the
external memory data structure needs storage comparable to the
 table and a (row-) compressed bitmap index.  This implementation
was initially used for our GC sorting, but it proved at
least 100 times slower than lexicographically sorting.
}
}

As an alternative to lexicographic sort and column
reordering, we introduce \emph{Frequent-Component} sorting, which
uses histograms to help sort
without bias from a fixed dimension ordering.
In sorting, we compare the frequency of the $i^\textrm{th}$ most frequent attribute values in 
each of two  rows without regard (except for possible tie-breaking)
to which columns they come from.
With appropriate pre- and post-processing, it is possible to implement
this approach using a standard sorting utility such as Unix \texttt{sort}.
\notessential{%
For two rows $r_1$ and $r_2$, $<_{\textrm{Frequent-Component}}$ first compares
 $(f(a_1,a_1))$ with $(f(a_2),a_2)$, where $a_1$ is the least
frequent  component in $r_1$ and $a_2$ is the least frequent 
component in $r_2$.  Note that $a_1$ and $a_2$ do not necessarily
come from the same dimension.  Ties are broken by the 
second-least-frequent components in $r_1$ and $r_2$, 
and so forth.

Frequent-Component can be implemented using a standard utility such as Unix
\texttt{sort}.  First, we sort the components of each row of the  table
into ascending frequency.  In this process, each component is
replaced by three consecutive components, $f(a)$, $a$, and $\textrm{pos}(a)$.
The third component records the column where $a$ was originally found.
Lexicographic sorting (via \texttt{sort} or a similar utility) of rows
follows, after which each row is put back to its original value
(by removing $f(a)$ and storing $a$ as component $\textrm{pos}(a)$).
}

\section{Picking the Right k-of-N}\label{sec:multi}

Choosing $k$ and $N$ are important
decisions.  We choose a single $k$ value for all dimensions\footnote{%
Except that for columns with small $n_i$, we automatically adjust
$k$ downward when it exceeds the limits noted at the end of 
Section~\ref{sec:bitmapIndexes}.}, leaving
the possibility of varying $k$ by dimension as future work.
Larger
values of $k$ typically lead to a smaller index and a faster construction
time---although we have 
observed cases where $k=2$ makes a larger index. 
However, query times increase
with $k$: there is a construction time/speed
tradeoff.

\paragraph{Larger $k$ makes queries slower}
We can bound the additional cost of queries. Assume ${L_i \choose k}=n_i$.
A given $k$-of-$L_i$ bitmap is the result of an OR operation over at most 
$kn_i/L_i\leq 3 n_i^{(k-1)/k}$~unary  
bitmaps.  Because 
$|\bigvee_i B_i|\leq \sum_i | B_i|$, 
the expected size of such a bitmap 
is no larger than  $3 n_i^{(k-1)/k}$~times the expected size of a unary bitmap. 
%
%
%
%
A query looking for one attribute value will have to AND together
$k$ of these denser bitmaps. 
The entire ANDing operation can be done
by $k-1$ pairwise ANDs that produce intermediate results whose EWAH sizes are
increasingly small:
$2k-1$~bitmaps are thus processed. 
Hence, the expected time complexity of  an equality query on a dimension of size
$n_i$ is no more than  $3(2k-1) n_i^{\frac{k-1}{k}}$~times higher than the expected
cost of the same query on a $k=1$ index.

For a less pessimistic estimate of this dependence, observe that indexes
seldom increase in size when $k$ grows.  We may conservatively assume that
 index size is unchanged when $k$ changes.
Therefore the expected size of one bitmap grows as $\approx n_i^{-1/k}/k$, 
leading to queries whose cost is proportional to 
$(2-1/k)n_i^{-1/k}$.  Relative to the cost for $k=1$, which
is proportional to $1/n_i$, we can say that 
increasing $k$ leads to
queries that are
                 $(2-1/k)n_i^{(k-1)/k}$ 
times more expensive than
on a simple bitmap index.

For example, suppose $n_i = 100$, 
going from
$k=1$ to $k=2$ should increase query cost about 15 fold 
but no more than 90 fold. 
%
In summary, the move from $k=1$ to anything larger can have a 
dramatic negative effect on query speeds.  Once we are at $k=2$,
the incremental cost of going to $k=3$, $k=4$ is not  so high:
whereas the ratio $k=2/k=1$ goes as $\sqrt{n_i}$, the ratio
$k=3/k=2$ goes as $n_i^{1/6}$.


\cut{%
We would, at least, have to read all these
bitmaps, and so we might expect equality query times to grow with $k^2$.  
However, the total size of the index may be smaller.
To understand this tradeoff, we focus on primarily on how $k$ affects 
the index size.  
\textbf{daniel analysis to replace above.}
{Let me redo the analysis from scratch. Make it clear
that this is for equality queries only.
Ok. We know that taking AND between two bitmaps can be done
in time proportional to the sum of the compressed sizes of the bitmaps.
For the AND operation between $k$ bitmaps, a pessimistic bound
tells you that the operation is in time proportional to $(2k-1)B$
where $B$ is the size of the bitmaps (here, we should state
precisely the algorithm we use: merge two compressed bitmaps into 
another compressed bitmaps and so on). This is pessimistic because
AND operations between bitmaps tends to generate smaller and smaller bitmaps.
 Ok, we have $k n_i^{1/k}$ bitmaps where $n_i$ is the number of
attributes. We know
that the total size of the bitmap index hardly ever increases with $k$ and
usually diminishes. Let us be pessimistic and assume that the size
remains the same as $k$ increases. Then the size of each bitmap
goes as $\frac{S}{k n_i^{1/k}}$. Hence, my pessimistic bound
puts the query time proportional to 
$\frac{(2k-1) S}{k n_i^{1/k}}$. So, for big dimensions, there might
be a huge difference between the various $k$, but less so for the smaller
dimensions. Of course, the size of the index $S$ diminishes with $k$
and the $(2k-1)B$ is likely to be quite pessimistic for $k=4$.
Suppose we are dealing with a small dimension $n_i=100$, the
the ratio of the query speed of $k=2$ to $k=1$ is 15
whereas $k=3/k=1$ is  35 and $k=4/k=1$ is 55. Increasing the dimension
size considerably, say $n_i=10000$, we get more drastic differences:
ratio $k=2/k=1$ is 150, $k=3/k=1$ is 774 and $k=4/k=1$ is 1750.
To test it out I ran the following experiment. I generated a 1d table
with a uniformly distributed histogram. I set the table length to 100,000
rows and I varied the number of distinct values and $k$. I then
plotted the time (in seconds) to execute 1,000 queries on one bitmap (alternatively,
you could say that the time reported is in millisec):\\
\includegraphics[width=0.8\columnwidth]{unisynthtimeversuscardi}\\
Then I plotted the $\frac{(2k-1) S}{nbrbitmaps}$ estimate (in words):\\
\includegraphics[width=0.8\columnwidth]{unisynthbitmapsizeversuscardi}\\
There is no agreement between the bound and the data,  but it is a pessimistic bound.
 }
\daniel{Owen properly reminds us that hardware and OS-level caching will help
 the
looped queries, thus making our results a bit unrealistic. This should
be pointed out in the text. }
\daniel{I could easily plot something like $\frac{k S}{nbrbitmaps}$ which
would give us the amount of data being loaded up. It might be false to
equate it with the query time, but it certainly contribute to the
query time.}
}

\cut{%
\owen{This is more the subject of one of the other manuscripts\ldots}
\owen{following paragraph is probably removable junk; DL has earlier put
some similar accounting.}
With $k=1$, there is a wide range on the number of dirty words:
Since the index will have $rc$ 1-bits, it can have at most $rc$
dirty words. (Ignoring really long runs of zeros).  In the first
dimension.  The EWAH cost can be bounded by observing that
the worst case would have a run of zeros between every dirty word
in a bitmap,
as well as before the first and last dirty word.  For $k=1$,
a lower bound is that each of the bits must be accounted for.
Ideally, each bitmap would consist of a run of zeros,
a run of ones, and a final run of zeros.  It would require
3 words in  EWAH\@.  More realistically, we
would not have the run of ones aligned on word boundaries, and
thus a dirty word would would follow and precede the run of ones,
leading to
5 words per bitmap.
Thus, for $k=1$, our index EWAH size is  between $3N$ and 
$\min(3N, N+ 2*rc)$.
Repeating for $k=2$, we have a total of $2rc$ bits, and an
index EWAH size between $3 \Theta(\sqrt{N})$ (fixme) and 
$\min(3\Theta(\sqrt{n}), \Theta(\sqrt{N} + 4*rc)$
Thus, we see that as $k$ increases our bounds become broader.
\owen{It is easy to construct inputs meeting the lower bound, but
what about the upper bound?}
}


\paragraph{Larger $k$ makes indexes smaller}
Consider the effect of a length 100 run of values $v_1$, followed
by 100 repetitions of $v_2$, then 100 of $v_3$, etc.  Regardless of
$k$, whenever we switch from $v_1$ to $v_{i+1}$ at least two bitmaps
will have to make transitions between 0 and 1. 
Thus, unless the
transition appears at a word boundary,
 we create at least 2
dirty words whenever an attribute changes from row to row.  The
best case, where \emph{only} 2 dirty words are created, is achieved when $k=1$
for \emph{any} assignment of bitmap codes to attribute values.
For $k>1$ and $N$ as small as possible, it may not be possible to
achieve so few dirty words, or it may require a particular assignment
of bitmap codes to values.


Encodings with $k>1$ find their use when many (e.g. 15) attribute
values fall within a word-length boundary.  In that case, a
$k=1$ index will have at least 15 bitmaps with transitions
(and we can anticipate 15 dirty words).  However, if there were only
45 possible values in the dimension, we would not need more than 10 bitmaps
with $k=2$.  Hence, there would be at most 10 dirty words and
maybe less if we have sorted the data (see Fig.~\ref{fig:adjdirty-thousand}).  
\cut{In
fact, with 45 values in a dimension encoded with 2-of-10 encoding,
a random selection of 15 distinct codes has an expected number of 0.35
bits in common (see Proposition~\ref{numofexpectedbits}).}

\cut{%
removing attempt to quantify (not very helpful, not very right)

\owen{this has a lot in common with things DL added on June 9,
so there is room to save space here.}

\begin{proposition}\label{numofexpectedbits}
Given $x$ distinct values with $k$-of-$N$ encodings, the expected number
of bits where all codes agree is $formula(x,k,N)$.
Something like $N ( (1-k/N)^x + (k/N)^x)$ but not exactly.
\end{proposition}

\begin{proof}
Depends on formula.  Sanity: 
formula should be $0$ when $x \geq {N-1 \choose k}$.
Should be $k$ when $x=1$.

The probability that any given bit position is clean is the
sum of the probability that it is a clean 1 and the probability
that it is a clean 0, since these possibilities exclude one
another.

In a code, suppose that each bit position has an independent
probability of being 1, and it is $k/N$.  (This is obviously wrong,
since if $k$ is 2 and you know two other bit positions are 1, then
the probability of any remaining bit position storing 1 is 0.  Oh
well, this is a starting point.  Also, I want $x$ distinct codes
(no repetition)).   The probability that in $x$
random codes, there are only zeros [in a pre-specified bit position]
 is $(1-k/N)^x$, and the probability
there are only ones is $(k/N)^x$.
\end{proof}

\paragraph{Comment:} For our case study, evaluating
$N ( (1-k/N)^x + (k/N)^x)$ with $N=10,\ k=2,\ x=15$ we have
0.35; in other words, we expect 9.65 dirty words rather than 15.
This 0.35 is not very significant.  For a fixed small value of
$k$, this wrong formula is significant wrt N when N is large and
x is fairly small.  Eg, x=4, N = 20, k=2 we have about $.9^4 20$
or about 13.  13 clean words, hence 7 dirty words.   However,
simple coding would have only 4 dirty words.
}
%

\paragraph{Choosing $N$}


It seems intuitive, having chosen $k$,
to choose $N$ to be as small as possible.   Yet, we have observed
cases where the resulting 2-of-$N$ indexes are much bigger
than 1-of-$N$ indexes.  Theoretically, this could be avoided 
if we allowed larger $N$, because one could aways append an additional
1 to every attribute's 1-of-$N$ code.  Since this would create one
more (clean) bitmap than the 1-of-$N$ index has, this
2-of-$N$ index would never be much larger than the 1-of-$N$ index.
So, if $N$ is unconstrained, we can see that there is never a
significant space advantage to choosing $k$ small.

Nevertheless, 
the main advantage of $k>1$ is 
 fewer bitmaps. 
We choose $N$ as small
as possible. 

\cut{%
\owen{you can obtain a k-of-N code from any 1-of-N by prepending k-1
columns of constant 1s.  Thus, if you do not constrain N, a k-of-N
code can always be almost as good (additive Theta(c*k) amount) as
good as a 1-of-N.  Question: if I constrain N to be as small as possible
(what we have been doing), can we make ourselves get a worse answer
for 2-of-N than 1-of-N?  I think so, from experiments.  Make a simple
manual example.}

\daniel{We have a funny k-of-N implementation where we pick a different
k per dimension when dimensions have too few attribute values. Daniel was
motivated by the following idea: if you have the same number of bitmaps, you
are better off with the lesser value of $k$.}\daniel{We could prepare a finer
analysis for DOLAP.}  this is in section 2

\owen{for
3-of-N coding and 4-of-N coding,  and -a weightgray, you obtain a smaller
index using  -K than -k on tweed6,3,4,52.}
\daniel{not necessarily faster performances.}
\owen{I am not sure if this is just a quirk of the heuristicness of
everything else that goes on.  In some cases, I think if you exchange the role of 1 and 0, what happens with -K and a ``too big'' value for the
dim size may be the (negation) of what happens with -k  and N-''too big''.
Since 1s and 0s are handled symmetrically in EWAH, this is not very
interesting.

We could look into this more, but I am not optimistic about this any more.
}
}

\section{Experimental Results}\label{sec:Experiment}

We present experiments to assess the effects of various
factors (choices of $k$, sorting approaches, dimension
orderings) in terms of EWAH index sizes.   These factors
also affect index creation
and query times (we report real wall-clock times).

\subsection{Platform} 
Our test programs\footnote{\url{http://code.google.com/p/lemurbitmapindex/}.}
were written in C++ and compiled by GNU
GCC~4.0.2 on an Apple Mac Pro with two double-core Intel
Xeon processors (2.66\,GHz) and 2\,GiB of RAM\@.
Lexicographic sorts of flat files were done using GNU coreutils 
sort version 6.9.
For all tests involving $k=1$, we used the sparse implementation
approached in Section~\ref{sec:compression} because without it, the Gray-Lex
index creation times were 20--100 times slower, depending on
the data set.


\subsection{Data sets used}\label{sec:datasets}



We primarily used four data sets, whose details are
summarized in Table~\ref{tab:caractDataSet}: Census-Income~\cite{KDDRepository}, DBGEN~\cite{DBGEN},
KJV-4grams, and Netflix~\cite{netflixprize}.
\cut{redundant: Census-Income has 42~dimensions and 
about 200~thousand facts.}%
 DBGEN is a synthetic data set, whereas 
\cut{redundant: The fact table for DBGEN has 12~million facts and 
16~dimensions.}%
KJV-4grams is a large list (including duplicates) of
4-tuples of words obtained from the
verses in the King James Bible~\cite{Gutenberg}, 
after stemming
with the Porter algorithm~\cite{275705} and removal of
stemmed words with three or fewer letters.
Occurrence of row $w_1, w_2, w_3, w_4$  indicates
that the first paragraph of a verse contains words $w_1$ through $w_4$, in this
order. 
This data is a scaled-up version of word co-occurrence cubes
used to study analogies in natural
language~\cite{TurneyML,KaserKeithLemire2006}.
Each of KJV-4grams' columns contains roughly 8~thousand distinct stemmed words.
The Netflix  table has \cut{redundant: over 100~million facts and its}%
4~dimensions:  UserID, MovieID, Date and Rating,
having cardinalities 5, 2~182, 17~770, and 480~189.
Details of how it was obtained from the data downloaded
are given elsewhere~\cite{bda08}.

\begin{table}
     \caption{Characteristics of data sets used.
    }\label{tab:caractDataSet}
    \centering
    \begin{tabular}{l|rrrr|} \cline{2-5}
     & rows & cols & $\sum_i n_i$ & size \\ \hline
    \multicolumn{1}{|l|}{\textbf{Census-Income}} & 199~523     & 42 & 103~419   & 99.1\,MB   \\ 
    \multicolumn{1}{|r|} {4-d projection}        & 199~523     & 4 & 102~609   & 2.96\,MB   \\
    \multicolumn{1}{|l|}{\textbf{DBGEN}}         & 13~977~980  & 16 & 4~411~936   & 1.5\,GB \\ 
    \multicolumn{1}{|r|} {4-d projection}        & 13~977~980  & 4 & 402~544   & 297\,MB  \\
    \multicolumn{1}{|l|}{\textbf{Netflix}}       & 100~480~507 & 4 & 500~146   & 2.61\,GB \\
    \multicolumn{1}{|l|}{\textbf{KJV-4grams}}    & 877~020~839 & 4 &  33~553   & 21.6\,GB  \\ 
\hline
    \end{tabular}

\end{table}

For some of our tests, we chose four dimensions with a wide range
of sizes.
For Census-Income, we chose \textit{age} ($d_1$), 
\textit{wage per hour} ($d_2$), \textit{dividends from stocks} ($d_3$)
and a numerical value\footnote{%
The associated metadata says this column should be a
10-valued migration code.} found in the $25^{\textrm{th}}$ position ($d_4$). Their respective cardinalities 
were 91, 1~240, 1~478 and 99~800.
For DBGEN, we selected dimensions of cardinality 7, 11, 2~526 and 400~000.
Dimensions are numbered by increasing
size: column 1 has fewer distinct values. \cut{to the user.  But internally, Owen's scripts still name
things the other way, until it is presentation time.  Crazy!}

\notessential{%
\daniel{I am worried that TWEED might be a bad data set because I found
that the gain behaved similarly for all $k$'s when $n$ is small.
We see a big difference between $k=1$ and the $k>1$ for large values of $n$.
See Fig.~\ref{fig:theorygain} to see how $k=1$ differs quite a bit from the
others.}
We also used the small terrorism
data set, TWEED~\cite{tweed}, which has approximately 11~thousand rows. 
We projected
columns  
52 (Target), 
4 (Country),
3 (Year) and
6 (Acting Group), whose cardinalities were
11, 16, 53 and 297, respectively.
\daniel{Any reason to present the dimensions
in reverse cardinality order?}\owen{good catch.}
}

\newpage
\subsection{Column Ordering}
\label{sec:exper-ordering}

\cut{%
Expr parts of Section 4.3 of BDA\@.  New plots like those emailed by Owen.
Can we suggest good rules for ordering?  k=1 seems to require different
rules.}

Fig.~\ref{fig:orders} 
shows the Gray-Lex index
sizes for each column ordering. 
The dimensions
of KJV-4grams are too similar for ordering to be interesting,
\altnotessential{and we have thus omitted them. }%
\notessential{and thus we 
used TWEED instead in our tests.  
\daniel{See above remark regarding why TWEED might be evil.}
\owen{If we can explain this, it is a good contribution.
It shows why toy data sets are misleading.  Plus, some people will
want to build bitmap indexes on small relations and we would not
want to mislead them.}
\owen{Another reason for wanting TWEED is that we cannot nicely
show 3 data sets.  For 2, we need an excuse why we don't show the
3rd (what are we hiding?)  We could revive DBLP, though}}%
 For small dimensions, the value of $k$ was lowered using
the heuristic presented in Section~\ref{sec:bitmapIndexes}.
Our results suggest that table-column reordering has a 
significant effect (40\%).  
This does not 
contradict the observation
by Canahuate et.\ al~\cite{pinarunpublished} that \cut{column}%
bitmap
 reordering 
 does not change the 
 size much. \cut{did not seem to be important, because they were
reordering bitmaps after encoding, rather than attributes that are
then encoded into an index.}%

The value of $k$ affects which ordering leads to
the smallest index: 
good orderings for  $k=1$ are frequently bad orderings for
$k>1$, and vice versa.
This is consistent with our earlier analysis (see Figs.~\ref{fig:theorygain} and~\ref{fig:toyorderings}).
\notessential{For TWEED, all values
of $k$ suggested similar orderings (for instance, the second
largest dimension should not go first), although they differed
somewhat about which orders were best, if the third-largest
dimension goes first. 
\daniel{This behavior of TWEED is consistent with what my theoretical
investigations reveal. I say we should drop TWEED and explain that
we are interested in large tables only.}
However, for the other data sets, 
}%
For Netflix and DBGEN, we have omitted $k=2$ for legibility:
it is inferior to $k=1$ for most orderings.

\cut{
Our earlier conclusions~\cite{bda08} about the correct ordering were
too simple, because they did not account for $k$----except to note that
when $k$ becomes large, we did not believe the issue of ordering would be
as important.  However, Fig.~\ref{fig:orders} clearly disproves this.
}

\begin{figure*}[!t]
\centering
  \subfigure[Census-Income]{\includegraphics[width=0.33\textwidth]{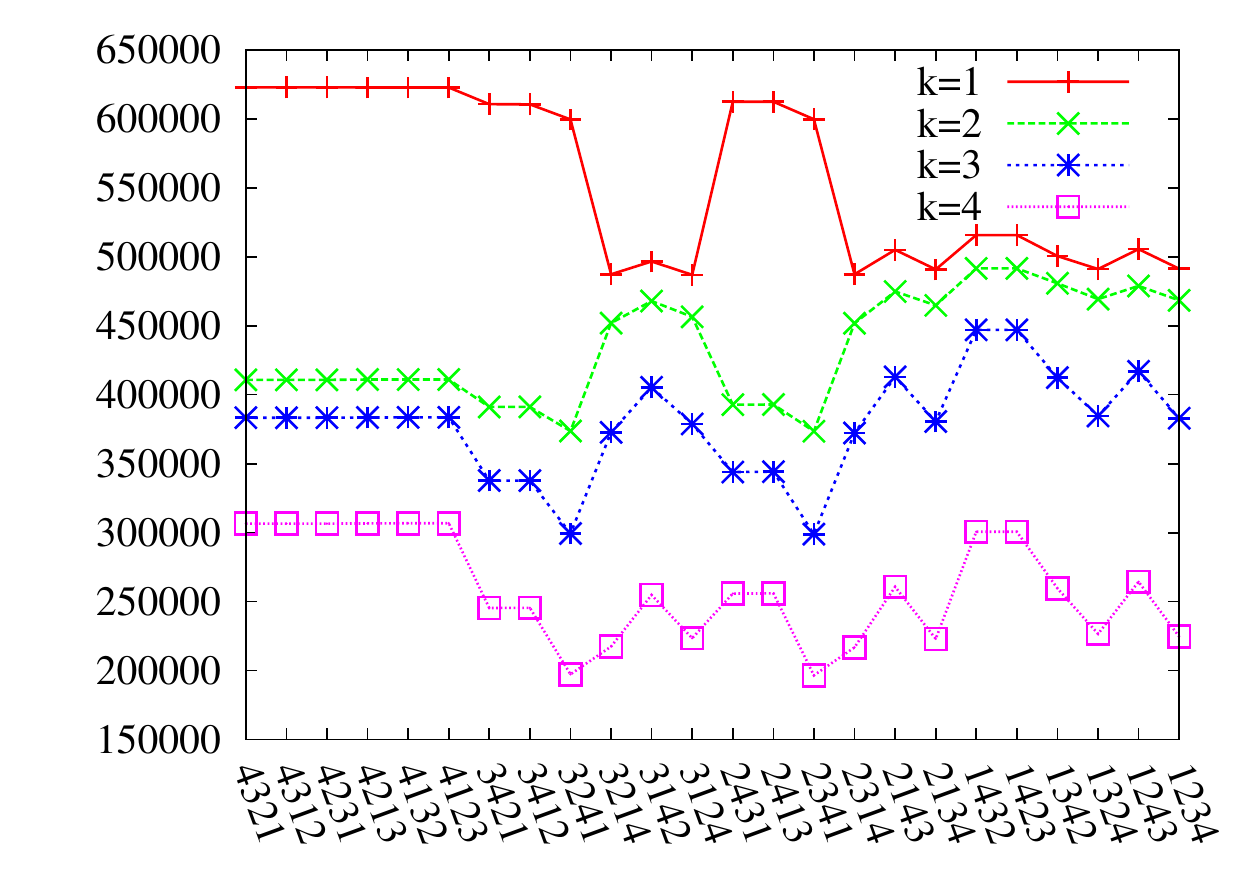}\label{fig:census-order}}
 \subfigure[DBGEN]{\includegraphics[width=0.33\textwidth]{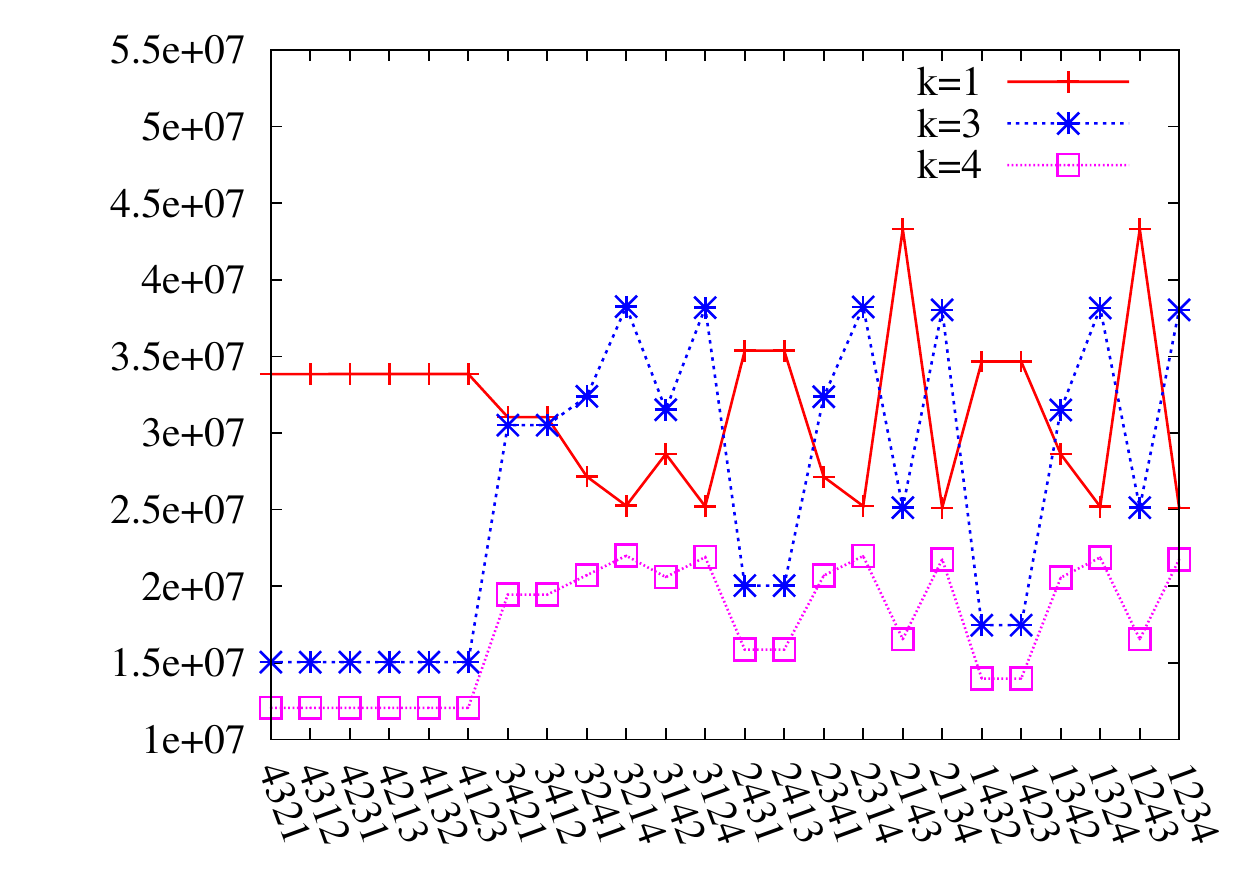}\label{fig:dbgen-orders}}
  \subfigure[Netflix]{\includegraphics[width=0.33\textwidth]{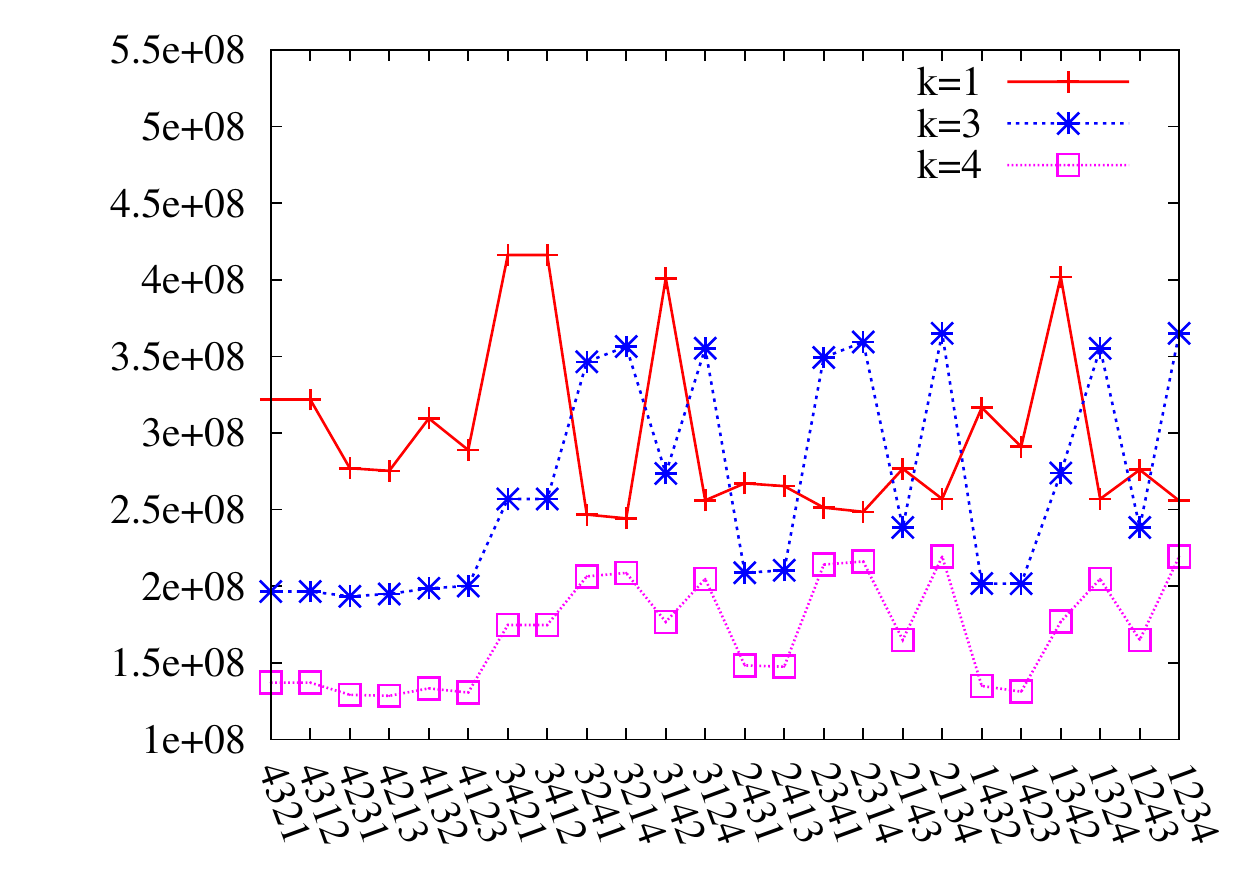}\label{fig:nf-order}}
\caption{Index sizes (words, $y$ axis) on 4-d data sets 
for all dimension orderings ($x$ axis).
\cut{now explanation has to come earlier
Ordering $abcd$ indicates the $a^{\textrm{th}}$ smallest dimension's
being the primary sort key, the $b^{\textrm{th}}$ smallest dimension's
being the secondary sort key, and so forth.}} \label{fig:orders}
\end{figure*}

\cut{%
\begin{figure*}[!t]
\centering
 \subfigure[Zipf 0.6-1.2]{\includegraphics[width=0.45\textwidth]{orderings-zipf2}\label{fig:zipfnarrow-order}}
 \subfigure[Zipf 0.4-1.6]{\includegraphics[width=0.45\textwidth]{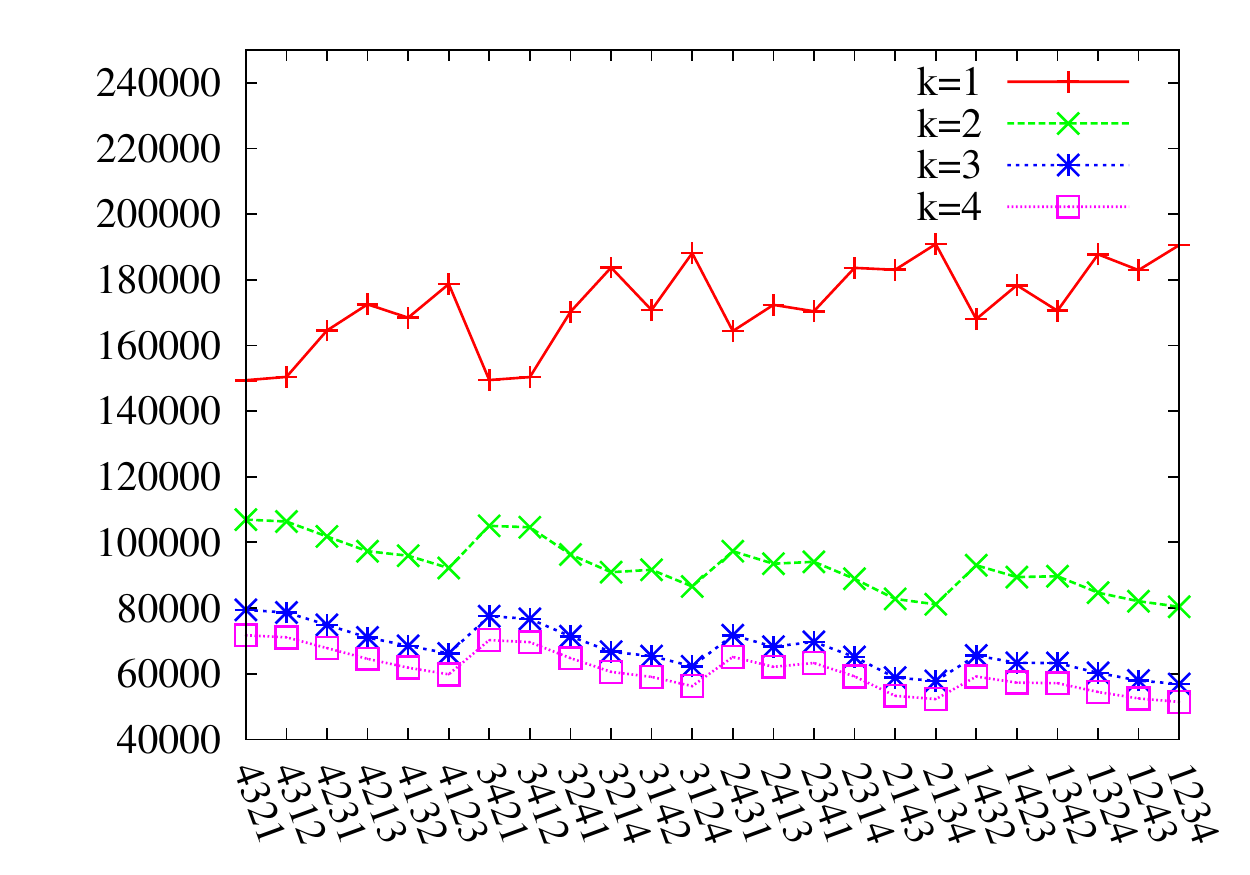}\label{fig:zipfwide-order}}
\\ \quad
  \subfigure[Zipf 0.2]{\includegraphics[width=0.45\textwidth]{orderings-zipf-size-s2}\label{fig:zipfs2-order}}
  \subfigure[Zipf 1.0]{\includegraphics[width=0.45\textwidth]{orderings-zipf-size-s10}\label{fig:zipfs10-order}}
\caption{Index sizes for various orders on synthetic data sets.
The \cut{top two have} first one has identical size dimensions with decreasing skew per dimension.
The \cut{bottom two have} second one has dimensions of size 10k, 1000, 100 and 10. Bottom left
has $s=0.2$ and right has $s=1.0$.  \owen{I know we may not
eventually have room for this.  A compromise might be one
skew picture and one dimsize?}} \label{fig:orders-synth}
\end{figure*}

Using synthetic data, we controlled skew and dimension sizes
separately. Fig.~\ref{fig:zipfnarrow-order}~and~\ref{fig:zipfwide-order} 
show index size for data sets with cardinality-100 dimensions.  However,
the dimensions were Zipfian with different amounts of skew (s=1.2,s=1.0,
s=0.8 and s=0.6 for the left figure; s=1.6,s=1.2,s=0.8 and s=0.4 for the
other). \cut{duplicat: For $k=1$, dimensions should be ordered from least skewed to most.
For $k>1$, the opposite ordering should be used.}

Fig.~\ref{fig:zipfs2-order}~and~\ref{fig:zipfs10-order} show
data sets with independent Zipfian dimensions with
skews $s=0.2$ and $s=1.0$, respectively. Both had dimension
cardinalities 10000, 1000, 100 and 10.  However, the high
cardinality dimensions had many values that appeared fewer than
32 times: of the 10k values there were 10 ($s=0.2$) or 324 ($s=1.0$)
occurring at least 32 times.   For $s=0.2$ and $k=1$, orderings 14xx
placing the two practically small dimensions  (1 and 4) first 
did poorly, but orderings 41xx did well.
Otherwise, ordering made little difference, as with
$s=1.0$ and $k=1$.  When $k>1$ best results are obtained by ordering
dimensions by descending cardinality, and the skew seemed less 
important.  It seems that the number of frequent attribute values
is important for $k=1$, but the total number of attributes matters
more for $k>1$. \owen{actually, this may make sense: the relationship of
a value with its bitmap is very straightforward for $k=1$, but for
$k>1$ the rare values interact with the ones that actually occur
frequently (they share bitmaps with them).  Does this somehow
explain anything>}

\cut{not really for real data, just this synth stuff
We see that differences in dimension size become less important as $k$
increases from 2: this may simply be that the number of bitmaps for
the various dimension becomes similar.
}
}


\cut{
\begin{table*}
\centering
{
\begin{tabular}{|lr|rrrrrr|}\cline{3-8}
	\multicolumn{2}{c|}{}   &                      &                      & \multicolumn{4}{c|}{\textbf{Ordering heuristic}} \\ 
	\multicolumn{2}{c|}{}   & \textbf{Smallest}    &   \textbf{Largest}   &  1         &      2  &          3  &        4\\ \hline
        TWEED                &k=1 &   4603             &      6109            &            & 4798    &  5738       &          \\
                             &k=2 &   3606             &      6289            &            & 3744    &  5500       &          \\
                             &k=3 &   3160             &      5528            &            & 3367    &  5093       &          \\
                             &k=4 &   3089             &      5125            &            & 3304    &  4798       &          \\ \hline
        DBGEN                &k=1 &    25121855        &   43307577           &            &33847445 &33847445     &          \\
                             &k=2 &    27639092        & 45230670             &            &27639093 &27639093     &          \\
                             &k=3 &    15038603        &38247513              &            &15038604 &15038604     &          \\
                             &k=4 &    12060789        &21983944              &            &12060794 &12060794     &          \\ \hline
        Census-Income        &k=1 &    486997          &  622943              &            & 487319  & 491073      &          \\
                             &k=2 &    373540          &  491597              &            & 451927  & 469307      &          \\
                             &k=3 &    298792          & 447027               &            & 372728  & 384607      &          \\
                             &k=4 &    196451          & 306712               &            & 217397  & 226520      &          \\ \hline
        Netflix              &k=1 &   244226400        & 416186602            &            &321829546&321829546    &          \\
                             &k=2 &   316826482        & 411942449            &            &316826482&316826482    &          \\
                             &k=3 &   193365286        & 365041817            &            &196575343&196575343    &          \\
                             &k=4 &   128451615        & 219329807            &            &137067789&137067789    &          \\ \hline
\end{tabular}                                                                                           
}                                                                                                       
\caption{\label{tab:dimension-ordering}Index sizes (words) with various orderings (4-d data sets).  Choice results
could be expressed in raw numbers or percentage worse than best.}
\end{table*}

See Table~\ref{tab:dimension-ordering} evaluates the success of the
proposed column ordering heuristics.
}

 Census-Income's largest dimension is very large ($n_4 \approx n/2$);
 DBGEN has also a large dimension ($n_4 \approx n/35$). Sorting columns
 in decreasing order with respect to  $\min(n_i^{-1/k}, (1-n_i^{-1/k})/(4w-1))$ for $k=1$,
 we have that only for DBGEN   the ordering
 ``2134'' is suggested, otherwise,  ``1234'' is recommended. Thus the heuristic
 provides nearly optimal recommendations. For $k=3$ and $k=4$, the ordering ``1234'' is recommended for
 all data sets: for $k=4$ and Census-Income, this recommendation is wrong.
 For $k=2$ and Census-Income, the ordering ``3214'' is recommended, another wrong 
 recommendation for this data set. Hence, a better column reordering heuristic is needed for $k>1$.
The difficulty appears to be fundamental
: when we calculated the gain experimentally,
we found that the best orderings sometimes did not have the dimensions
with highest gain first.  Our greedy approach may be too simple, and it
it may be necessary to know the histogram skews.

\cut{%
\owen{This discussion can be replaced by the table?}
For TWEED, this is
not the best ordering for any value of $k$, but it is never much
worse than the optimal ordering.   In TWEED, it is always a bad
idea for $d_3$ (Year) to go first.  This dimension's histogram 
\cut{is
found in Fig.~\ref{fig:tweed-year-histo}; we see that there
are}
has
a moderately large number of values, none of which are overwhelmingly
common.  In contrast, the numerically larger dimension $d_4$ has a 
much more skewed histogram and the top 4 values account for more
than half the records.
\cut{histogram shown in Fig.~\ref{fig:tweed-agent-histo}.}  
For Census-Income, the suggested
order is good for $k=1$, but it is a poor choice for $k>1$, where
either ordering 3241 or 2341 would be better.  The overall idea to
avoid orderings with the excessively large dimension ($d_4$) first is good,
however: the first 6 orderings for Census-Income are poor choices for any $k$.
For DBGEN, we hide $k=2$, because for most orderings it consumed
more space than $k=1$.  
For Netflix, the suggested ordering is good (but not optimum) for
$k>1$.  However, it is a poor choice for $k=1$, when the opposite ordering
``1234'' would be a much better choice.


When we counted an attribute value only if it occurs more than 32 times,
the suggested ordering for Netflix and DBGEN is (still) ``4321,''
\notessential{for TWEED it is 3421 (although 4321 is a near second), }
for Census-Income is  1324.
This is a poor choice for \notessential{TWEED and} Netflix, a good choice for $k=1$ or
$k=4$ on Census-Index, but a poor choice for $k=2$ or $k=3$. 
For DBGEN, it is a good choice for $k>1$.

For $k=1$ on Netflix, we see that some orderings beginning with $d_3$
(specifically, 3214 and 3241)
are among the best, whereas  others (3412 and 3421) are especially bad.
Since dimension 1 (rating) is tiny, functionally dependent on
the the other dimension,  and there are many occurrences of each value,
it may be sensible to ignore this value.    Then, we see that ordering
324 is good, but 342 is bad.
\owen{badly needs an explanation.  Maybe, tuples of (3,2) are nicely 
distributed?}
DBGEN, a synthetic data set, performed well for $k>1$ as long as the
largest dimension was first: the ordering of the other dimensions did
not matter.  However, for $k=1$, the reversed ordering, ``1234,'' was
among the best.

\owen{Netflix k=1.  34xx is bad but 32xx is good.  Similarly, DBGEN 
1234/2143 good but 1243/2143 bad. Weird?}
}

\subsection{Sorting}

On some synthetic Zipfian tests, we found
a small improvement 
(less than 4\% for 2~dimensions) by using Gray-Lex coding in preference
to Alpha-Lex~\cite[Fig.~3]{bda08}.  On other data sets, Gray-Lex
either had no effect or a small positive effect.   Therefore, our
current experiments  do not include Alpha-Lex, with the exception
that
we experimentally evaluated how sorting affects the EWAH
compression of individual columns.  Whereas
sorting tends to create runs of identical values in the first
columns,  the benefits of sorting are far less
apparent in later columns, except those \cut{that are} strongly
correlated with the first few columns.  
For Table~\ref{tab:sizesortcolumnorderd10},
 we have sorted
projections of Census-Income and DBGEN onto 10~dimensions
$d_1\ldots{}d_{10}$ with $n_1<\ldots<n_{10}$.  
(The
dimensions $d_1 \ldots d_4$ in this group are different from 
the dimensions $d_1 \ldots d_4$ discussed earlier.)
We see that if we sort from the largest column ($d_{10}\ldots{}d_1$),
at most 3 columns benefit from the sort, whereas 5 or more columns benefit
when sorting from the smallest column ($d_1\ldots{}d_{10}$).
{%
\begin{table*}
\caption{Number of 32-bit words used for different unary indexes when the
table was sorted lexicographically (dimensions ordered
by descending cardinality, $d_{10}\ldots d_1$, or by ascending
cardinality, $d_1\ldots d_{10}$).
}\label{tab:sizesortcolumnorderd10}
\centering
\begin{tabular}{|c|rrrr|rrrr|}\cline{2-9}
	\multicolumn{1}{c|}{}  & \multicolumn{4}{c|}{\textbf{Census-Income}}  & \multicolumn{4}{c|}{\textbf{DBGEN}} \\ 
	\multicolumn{1}{c|}{} & cardinality &unsorted &\multicolumn{1}{c}{$d_1 \ldots d_{10}$} & \multicolumn{1}{c|}{$d_{10} \ldots d_1$} & cardinality & unsorted & \multicolumn{1}{c}{$d_1 \ldots d_{10}$} & \multicolumn{1}{c|}{$d_{10} \ldots d_1$} \\ \hline
$d_1$ & 7 & 42~427 & 32 & 42~309 & 2 & 0.75$\times 10^6$ & 24 & 0.75$\times 10^6$ \\
$d_2$ & 8 & 36~980 & 200 & 36~521 & 3 &  1.11$\times 10^6$ & 38 & 1.11$\times 10^6$ \\
$d_3$ & 10 & 34~257 & 1~215 & 28~975 & 7 & 2.58$\times 10^6$ & 150 & 2.78$\times 10^6$ \\
$d_4$ & 47 & 0.13$\times 10^6$ & 12~118 & 0.13$\times 10^6$ & 9 & 0.37$\times 10^6$ & 100~6 & 3.37$\times 10^6$ \\
$d_5$ & 51 & 35~203 & 17~789 & 28~803 & 11 & 4.11$\times 10^6$ & 10~824 & 4.11$\times 10^6$ \\
$d_6$ & 91 & 0.27$\times 10^6$ & 75~065 & 0.25$\times 10^6$ & 50 & 13.60$\times 10^6$ & 0.44$\times 10^6$ & 1.42$\times 10^6$ \\
$d_7$ & 113 & 12~199 &  9~217 & 12~178 & 2~526 & 23.69$\times 10^6$ & 22.41$\times 10^6$ & 23.69$\times 10^6$ \\
$d_8$ & 132 & 20~028 & 14~062 & 19~917 & 20~000 & 24.00$\times 10^6$ & 24.00$\times 10^6$ & 22.12$\times 10^6$ \\
$d_9$ & 1~240 & 29~223 & 24~313 & 28~673 & 400~000 & 24.84$\times 10^6$ & 24.84$\times 10^6$ & 19.14$\times 10^6$ \\
$d_{10}$ & 99~800 & 0.50$\times 10^6$ & 0.48$\times 10^6$ & 0.30$\times 10^6$ & 984~297 & 27.36$\times 10^6$ & 27.31$\times 10^6$ & 0.88$\times 10^6$ \\ \hline
total  & - & 1.11$\times 10^6$ & 0.64$\times 10^6$ & 0.87$\times 10^6$  & - & 0.122$\times 10^9$ & 0.099$\times 10^9$ & $0.079\times 10^9$ \\ \hline
\end{tabular}
\end{table*}
}


\paragraph{Lexicographic sorting}

\notessential{%
Fig.~\ref{fig:kjv-size-vs-lines}~and~\ref{fig:kjv-time-vs-lines} 
show the sizes of simple bitmap indexes for various prefixes of KJV-4grams. 
It also shows the time to create these indexes.
\owen{It might be easier to grok as a log-log plot.  I ended up looking
at a nolog-nolog and log-log version while trying to understand it}
We see that the full index is nine times bigger if unsorted.  We also
see the sorted index appears to grow sublinearly, probably due to repeated
or almost repeated rows.  It should be faster to query a smaller index.
However, the time to sort and then create this smaller index is more than
the time to create the unsorted index.  The additional time is less than
the time to sort, and all times seem to scale linearly with the number
of rows. \daniel{We might stress here that this is external-memory sorting
since the machine has 2\,GiB of RAM\@ which is less than the size of the data set.}
}

Constructing a simple bitmap index (using Gray-Lex) over KJV-4grams took
approximately 14,000~seconds or less than four hours.  
Nearly half (6,000\,s) of the time
was due to the \texttt{sort} utility, since the data set is much larger 
than the machine's main memory (2\,GiB).  Constructing an unsorted 
index is faster (approximately 10,000\,s), but the index is about 9~times
larger.  

To study scaling, we built indexes from prefixes
of the full data set.  We found construction times increased linearly
with index size for $k=1$, whether or not sorting was used.  For 
$1\leq k \leq 4$, index size increased linearly with the prefix size for
unsorted data.  Yet with sorting, index size increased sublinearly.
As new data arrives, it is increasingly likely to fit into existing
runs, once sorted.   

Table~\ref{tab:index-sizes} shows index sizes for our large data sets,
using  Gray-Lex orderings\notessential{, Frequent-Component}  and Gray-Frequency. 
\notessential{For Gray-Lex and Gray-Frequency,  dimensions }%
\altnotessential{Dimensions }%
 were ordered from the largest to the smallest (``4321'')
except for Census-Income where we used the ordering ``3214''.
We observed that KJV-4grams did not benefit
in index size for $k=2$.
This data set has many very long runs of identical attribute values
in the first two dimensions, and the number of
attribute values is modest, compared with the number of rows.
This is ideal for 1-of-$N$.

Gray-Frequency yields the  smallest indexes in Table~\ref{tab:index-sizes}.
Frequent-Component is not shown in the table.  On Netflix for $k=1$ it
outperformed the other approaches by 1\%, and for DBGEN it was only slightly
worse than the others. But in all other case on DBGEN, Census-Income and Netflix,
it lead to indexes 5--50\% larger. 

\notessential{%
\begin{figure}
\centering
\includegraphics[width=.49\textwidth]{sizevsnboflines-alpha}
\caption{\label{fig:kjv-size-vs-lines} Bitmap index size for various
prefixes of the KJV-4grams table, $k=1$, shuffled table
and lexicographic table order.}
\end{figure}

\begin{figure}
\centering
\includegraphics[width=.49\textwidth]{timevsnboflines-alpha}
\caption{\label{fig:kjv-time-vs-lines} Times to create an index,
for various prefixes of the KJV-4grams table. $k=1$}
\end{figure}

\begin{figure}
\centering
\includegraphics[width=.49\textwidth]{sizevsnboflines-grayall}
\caption{\label{fig:kjv-size-vs-lines-gray} Bitmap index size for various
prefixes of the KJV-4grams table, $k>1$, shuffled table
and Gray code ordering.}
\end{figure}
}

\begin{table}
\caption{ \label{tab:index-sizes} Sizes of EWAH indexes (32-bit words) for various sorting methods.}
\centering\small
\begin{tabular}{|cr|rrr\notessential{r}|}\cline{3-5} 
\multicolumn{1}{c}{} & 
\multicolumn{1}{c|}{} & Lex unsorted & Gray-Lex & Gray-Freq\notessential{& Freq-Component} \\ \hline
Census-Income &  $k=1$ & $8.49\times10^5$  &\cut{  487319  }$4.87\times10^5$& \cut{487464   }$4.87\times10^5$\notessential{& 617381    }    \\
(4d)  &      2    & $9.12\times10^5$   &\cut{  451927  }$4.52\times10^5$& \cut{435595   }$4.36\times10^5$\notessential{& 426312    }    \\
      &      3    & $6.90\times10^5$   &\cut{  372728  }$3.73\times10^5$& \cut{327939   }$3.28\times10^5$\notessential{& 377165    }    \\
      &      4    & $4.58\times10^5$   &\cut{  217397  }$2.17\times10^5$& \cut{197695   }$1.98\times10^5$\notessential{& 329624    }    \\ \hline
DBGEN &      1  &  $5.48\times10^7$  &\cut{  33847445}$3.38\times10^7$& \cut{33847116 }$3.38\times10^7$\notessential{&  33851895}     \\
(4d)  &      2    &  $7.13\times10^7$  &\cut{27639093  }$2.76\times10^7$& \cut{27441584 }$2.74\times10^7$\notessential{&  27816958}     \\
      &      3    &  $5.25\times10^7$  &\cut{15038604  }$1.50\times10^7$& \cut{14977582 }$1.50\times10^7$\notessential{&  15542383}     \\
      &      4    &   $3.24\times10^7$ &\cut{ 12060794 }$1.21\times10^7$& \cut{11906432 }$1.19\times10^7$\notessential{&  12877534}     \\ \hline
Netflix &      1   &\cut{619959382} $6.20\times10^8$   &\cut{321829546 }$3.22\times10^8$& \cut{318959993}$3.19\times10^8$\notessential{& 313900325 }    \\
      &      2    & \cut{827496560} $8.27\times10^8$   &\cut{316826482 }$3.17\times10^8$& \cut{242676792}$2.43\times10^8$\notessential{& 388117623 }    \\
      &      3    &  \cut{573296926} $5.73\times10^8$  &\cut{196575343 }$1.97\times10^8$& \cut{149062457}$1.49\times10^8$\notessential{& 226287832 }    \\
      &      4     & \cut{341693411}$3.42\times10^8$  &\cut{137067789 }$1.37\times10^8$& \cut{113820043}$1.14\times10^8$\notessential{& 152122246 }    \\ \hline
KJV-4grams &      1 & $6.08\times10^9$ &\cut{ 668244925}$6.68\times10^8$& \cut{667570961}$6.68\times10^8$\notessential{& n/a       }    \\
      &      2    &  $8.02\times10^9$  &\cut{ 992824611}$9.93\times10^8$& \cut{729015073}$7.29\times10^8$\notessential{& n/a       }    \\
      &      3    & $4.13\times10^9$    &\cut{ 830621456}$8.31\times10^8$& \cut{577319979}$5.77\times10^8$\notessential{& n/a       }    \\
      &      4    &    $2.52\times10^9$ &\cut{ 638934508}$6.39\times10^8$& \cut{500834200}$5.01\times10^8$\notessential{& n/a       }    \\ \hline
\end{tabular}
\end{table}

\cut{
\paragraph{Histogram-Oblivious Interleaving}

When tested, the scheme for interleaving bitmaps (see the end of 
Section~\ref{sec:dim-order}) worked poorly.
To avoid issues with from differing
sized dimensions, we experimented with synthetic cubes with
91 attributes per dimension distributed with Zipfian ($s=1.0$)
and 5000 rows (apprx 4000 distinct rows). 
A Gray-Lex bitmap index was generated, then bitmaps
were interleaved, and finally the result was GC sorted.

For $k=3$ and $d=3$, interleaving increased the size of
the index by an average of  5.4\% (taken every 1000 cubes). 
For $k=2$ and $d=3$, the average index size increased
by 5.3\%.
\notessential{
On TWEED, we saw a 44\% size increase \owen{hmm which k}
from interleaving, compared with Gray-Lex with largest dimension
first.}
}
\cut{
The Frequent-Component approach seemed to work poorly 
(see Table~\ref{tab:index-sizes}).  Owen thinks he may
have run the tests incorrectly because he was not expecting this.
\daniel{I got the same type of results. Canahuate et al.~\cite{pinarunpublished}
reported that their interleaving was not very good at
increasing the compression, so we just confirm their results.
Also, the type of interleaving we do is akin to $k=1$ bitmap
reordering. There is no reason to believe it would work well
for $k>1$. So the fact that it does not make things too much
worse for $k=1$, but make things worse for $k>1$ should not
come as a surprise. }
}

\subsection{Queries}

%
We timed equality queries against our 4-d bitmap indexes, and the
results are shown in Fig.~\ref{fig:queries}.
Queries were generated by choosing attribute values uniformly
at random and the figures report average times for such queries.
We made 100 random choices per column
for KJV-4grams when $k>1$.
For DBGEN and Netflix, we had 1,000 random choices per column
and 10,000 random choices were used for Census-Income
and KJV-4grams ($k=1$). For each 
 data set, 
we give the results per column (leftmost tick
 is the 
column used as the primary sort key, next 
tick
is for the
secondary sort key, etc.).

From Fig.~\ref{fig:randqueries-times}, we see that simple bitmap
indexes always yield the fastest queries.  The difference caused
by $k$ is highly dependent upon the data set and the particular column
in the data set. However, for a given data set and column, with only
a few small exceptions, query times increase significantly with $k$.
For DBGEN, the last two dimensions have size 7 and 11, whereas
for Netflix, the last dimension has size 5, and therefore,
they will never use a $k$-value larger than 2: their speed is mostly
oblivious to 
$k$.

In Section~\ref{sec:multi}, we
predicted that the query time would grow with $k$ as $\approx (2-1/k)n_i^{-1/k}$:
for the large dimensions such as the largest ones for DBGEN (400k)
and Netflix (480k), query times are two orders of magnitude slower for $k=2$ as opposed
to $k=1$, and four orders of magnitude slower for $k=4$. Thus, our model exaggerates
 the differences
by about an order of magnitude. The most plausible explanation is that
query times are not directly proportional to the bitmap loaded, but also include
a constant factor.

%

Fig.~\ref{fig:randqueries-times-unsorted} and~\ref{fig:randqueries-times}
show the equality query times per column before and after sorting the tables.
Sorting improves query times most for larger values of $k$: for Netflix, sorting improved
the query times by  at most 2 for $k=1$, 
at most 40  for $k=2$ and at most 140 for $k=3$; indexes with $k>1$ benefit
from sorting even when there are no long runs of identical values (see Subsection~\ref{sec:sorting-rows}).  (On the first
columns, $k=3$ usually gets the best improvements from sorting.)
\owen{Or can we appeal to the model?}
Synthetic DBGEN showed no
significant speedup from sorting, beyond its large first column.  Although
Netflix, like DBGEN, has a many-valued column first, it shows a benefit
from sorting even in its third column: in fact, the third column benefits
more from sorting than the second column.  \owen{skew? correlations?}
The largest table, KJV-4grams, benefited most from the sort: while queries
on the last column are twice as fast, the gain on the first two columns 
ranges from 20 times faster ($k=1$) to almost 1500 times faster ($k=3$).
\owen{using
\includegraphics[width=0.4\textwidth]{randqueries-ratio} 
}%


We can compare these times with the expected amount of data scanned per query.
This is shown in Figure~\ref{fig:queries-sizes}, and we observe 
reasonably close
agreement 
\owen{see\\\includegraphics[width=0.4\textwidth]{querysizetime-srt})}%
between most query times and the expected sizes of the bitmaps 
being scanned. 
Exceptions include the first dimension on KJV-4grams and some cases
where the bitmaps are tiny. This discrepancy might be explained by the 
retrieval of the row IDs from the compressed bitmaps: long runs of 1x11
clean words must be converted to many row IDs. \daniel{My 
code retrieves the row IDs by iterating over the non-zero words. In the
case where we have compressed runs of 1x11, I expect that the size of the
bitmap is no longer a good predictor of the query time.}

\begin{figure*}[!t]
\centering
 \subfigure[Query times over unsorted \cut{(Gray-Lex)} indexes]{\includegraphics[width=0.8\columnwidth]{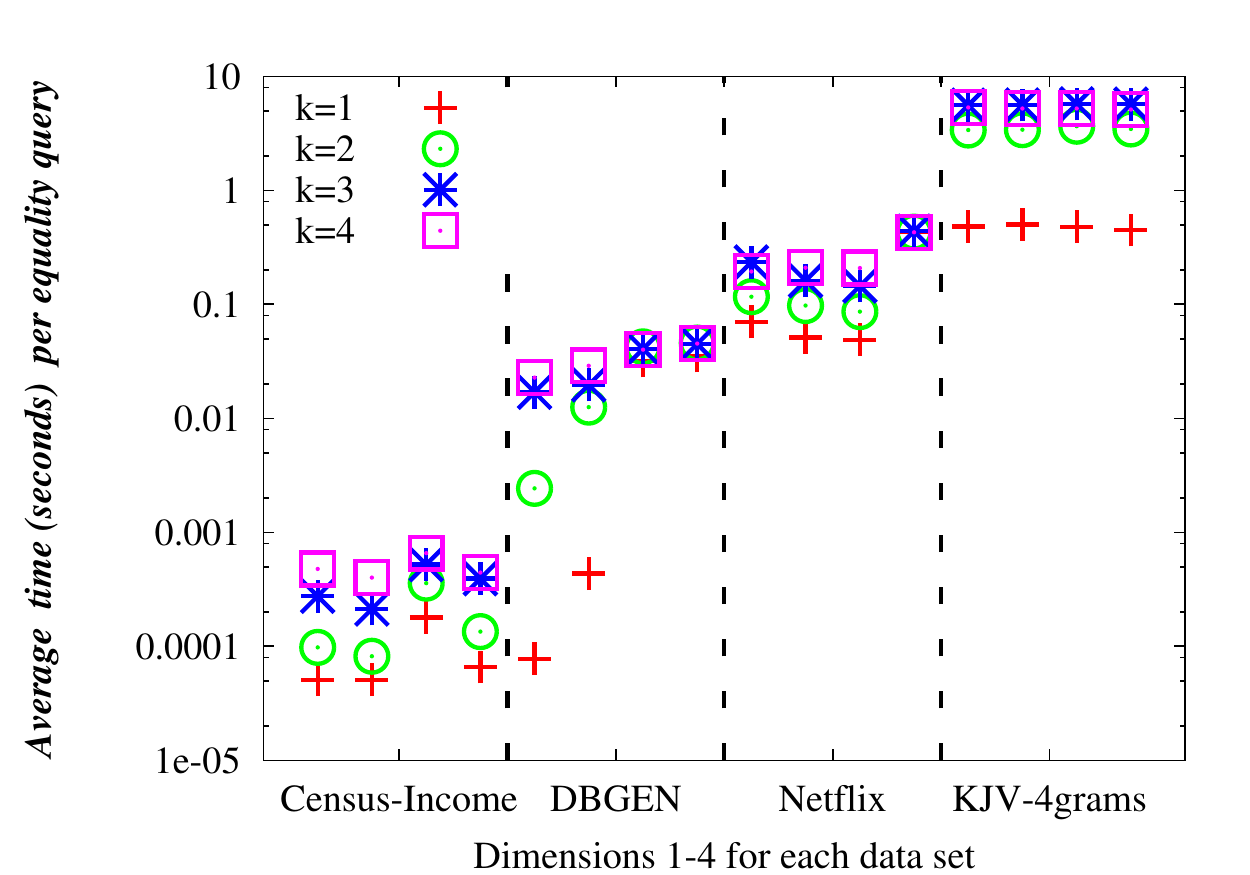}\label{fig:randqueries-times-unsorted}}
 \subfigure[Query times over sorted (Gray-Lex) indexes]{\includegraphics[width=0.8\columnwidth]{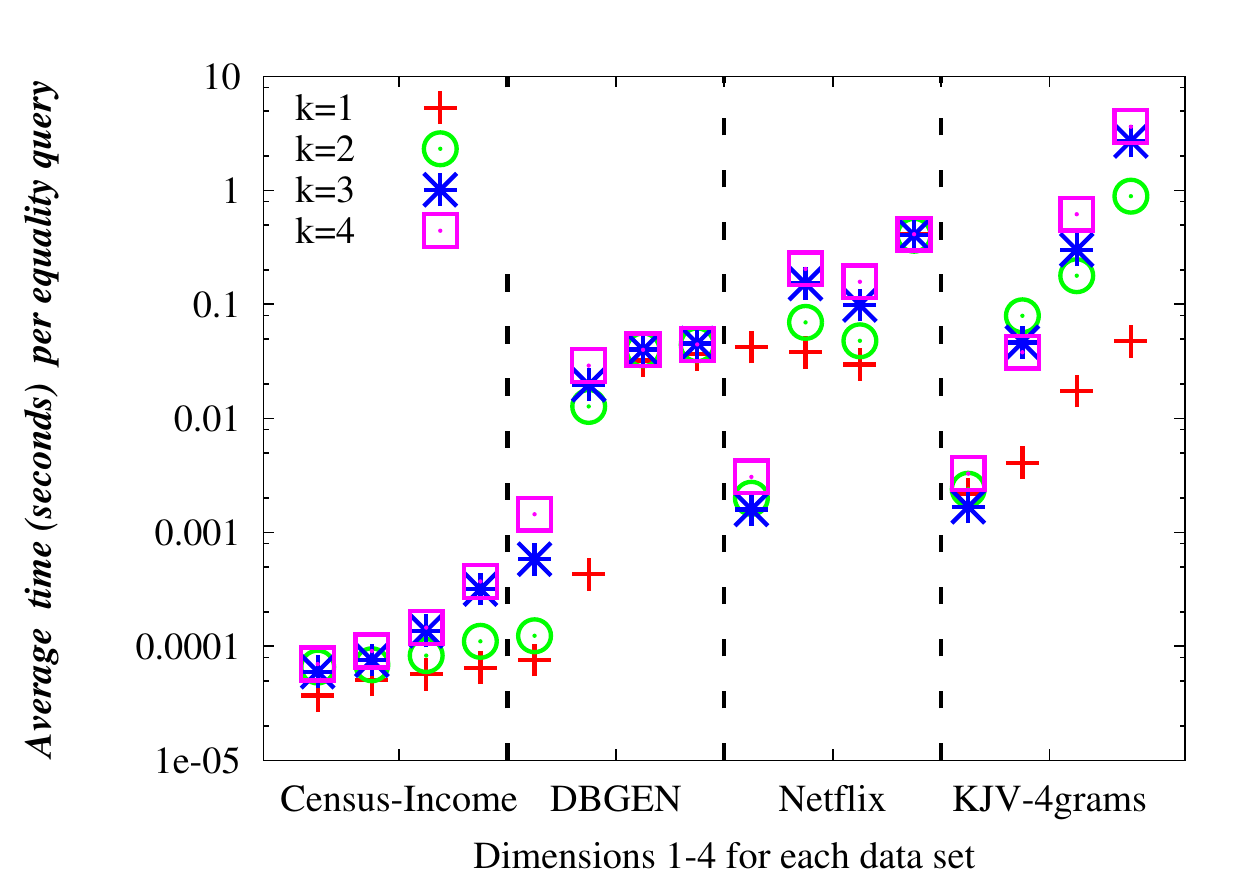}\label{fig:randqueries-times}}
\caption{Query times are affected by dimension, table sorting and $k$.
}\label{fig:queries}
\end{figure*}

\begin{figure} 
\centering
{\includegraphics[width=0.8\columnwidth]{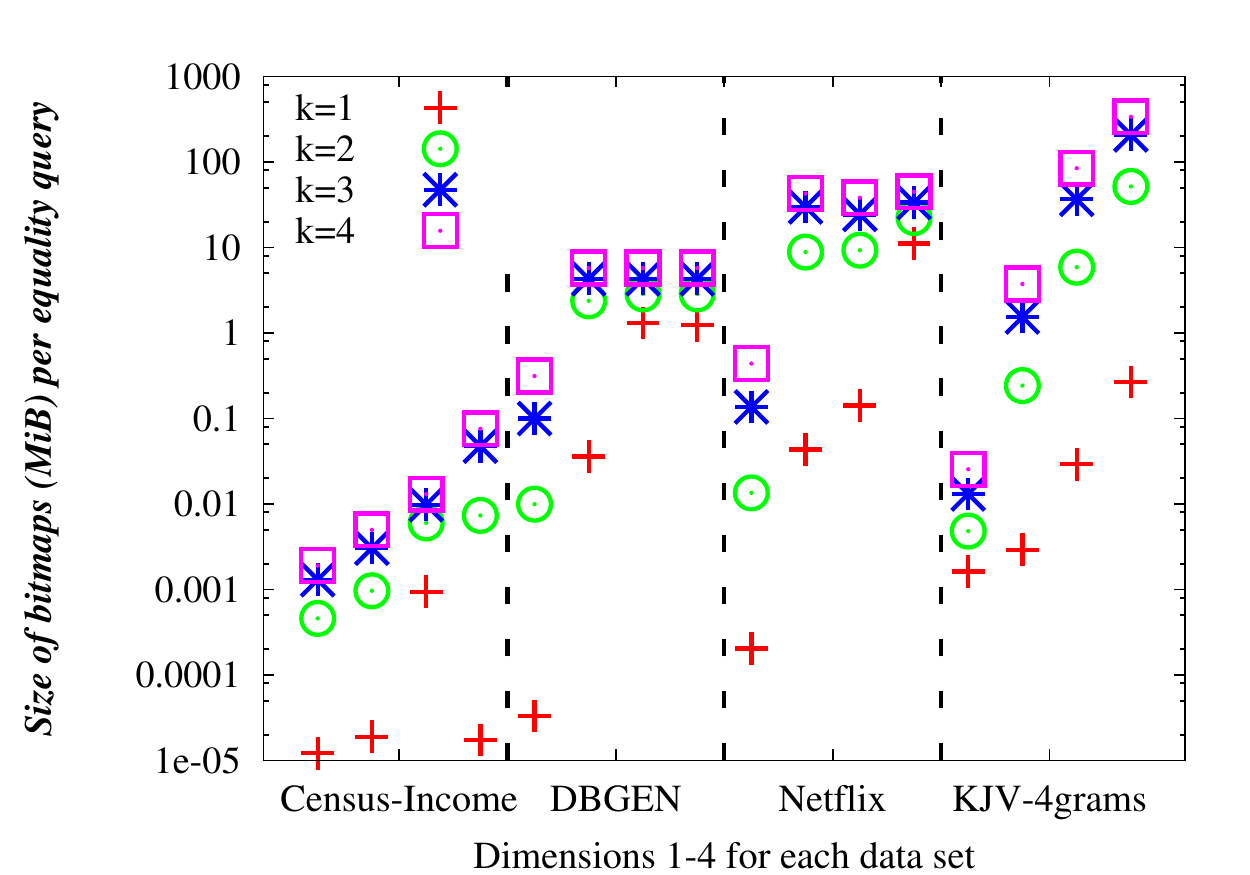}
}
\caption{Bitmap data examined per equality query.
}\label{fig:queries-sizes}
\end{figure}

\section{Guidelines for k}

Our experiments indicate that
simple ($k=1$) bitmap encoding is preferable when storage space and
index-creation time are less important than fast equality queries. 
 The storage and index-creation penalties are kept
modest by table sorting and Algorithm~\ref{algo:owengenbitmap}.

Space requirements can be reduced by choosing $k>1$,
although Tab.~\ref{tab:index-sizes} shows that this approach has risks
(see KJV-4grams). For $k>1$,
we can gain additional index
size reduction 
at the cost of longer index construction 
by using Gray-Frequency rather than Gray-Lex.  \notessential{The amount of space
gain is difficult to predict. }

If the total number of attribute values is
small relative to
the number of rows, then we should 
first
try the $k=1$ index.  Perhaps the data set resembles
KJV-4grams.  Besides yielding faster queries, 
the $k=1$ index may be smaller.


\section{Conclusion and Future Work}\label{sec:Conclusion}

We showed that while sorting improves bitmap indexes,
we can improve them even more (30--40\%) if we know the number of distinct values
in each column. For $k$-of-$N$ encodings with $k>1$, even further gains (10--30\%) are
possible using the frequency of each value. 
Regarding future work, 
the accurate mathematical modelling of compressed bitmap indexes
 remains an open problem.



\section*{Acknowledgements}
This work is supported by NSERC grants 155967, 261437 and by  FQRNT 
grant 112381.  
\balance
\bibliographystyle{abbrv}

\begin{thebibliography}{10}

\bibitem{874730}
G.~Antoshenkov.
\newblock Byte-aligned bitmap compression.
\newblock In {\em DCC '95}, page 476, 1995.

\bibitem{bda08}
K.~Aouiche, D.~Lemire, and O.~Kaser.
\newblock Tri de la table de faits et compression des index bitmaps avec
  alignement sur les mots.
\newblock available from \url{http://arxiv.org/abs/0805.3339}.

\bibitem{bellatreche2007sap}
L.~Bellatreche, R.~Missaoui, H.~Necir, and H.~Drias.
\newblock Selection and pruning algorithms for bitmap index selection problem
  using data mining.
\newblock {\em LNCS}, 4654:221, 2007.

\bibitem{pinarunpublished}
G.~Canahuate, H.~Ferhatosmanoglu, and A.~Pinar.
\newblock Improving bitmap index compression by data reorganization.
\newblock \url{http://hpcrd.lbl.gov/~apinar/papers/TKDE06.pdf}~(checked
  2008-05-30), 2006.

\bibitem{chan1998bid}
C.~Y. Chan and Y.~E. Ioannidis.
\newblock Bitmap index design and evaluation.
\newblock In {\em SIGMOD'98}, pages 355--366, 1998.

\bibitem{chan1999ebe}
C.~Y. Chan and Y.~E. Ioannidis.
\newblock An efficient bitmap encoding scheme for selection queries.
\newblock In {\em SIGMOD'99}, pages 215--226, 1999.

\bibitem{1183529}
R.~Darira, K.~C. Davis, and J.~{Grommon-Litton}.
\newblock Heuristic design of property maps.
\newblock In {\em DOLAP'06}, pages 91--98, 2006.

\bibitem{davis2007idw}
K.~Davis and A.~Gupta.
\newblock {\em Data Warehouses and {OLAP}: Concepts, Architectures, and
  Solutions}, chapter Indexing in Data Warehouses.
\newblock IRM Press, 2007.

\bibitem{KDDRepository}
S.~Hettich and S.~D. Bay.
\newblock The {UCI} {KDD} archive.
\newblock \url{http://kdd.ics.uci.edu}~(checked 2008-04-28), 2000.

\bibitem{qdbm}
M.~Hirabayashi.
\newblock {QDBM}: Quick database manager.
\newblock \url{http://qdbm.sourceforge.net/}~(checked 2008-02-22), 2006.

\bibitem{KaserKeithLemire2006}
O.~Kaser, S.~Keith, and D.~Lemire.
\newblock The {LitOLAP} project: Data warehousing with literature.
\newblock In {\em CaSTA'06}, 2006.

\bibitem{354819}
N.~Koudas.
\newblock Space efficient bitmap indexing.
\newblock In {\em CIKM '00}, pages 194--201, 2000.

\bibitem{netflixprize}
{Netflix, Inc.}
\newblock Nexflix prize.
\newblock \url{http://www.netflixprize.com} (checked 2008-04-28), 2007.

\bibitem{pinar05}
A.~Pinar, T.~Tao, and H.~Ferhatosmanoglu.
\newblock Compressing bitmap indices by data reorganization.
\newblock In {\em ICDE'05}, pages 310--321, 2005.

\bibitem{275705}
M.~F. Porter.
\newblock An algorithm for suffix stripping.
\newblock In {\em Readings in information retrieval}, pages 313--316. Morgan
  Kaufmann, 1997.

\bibitem{Gutenberg}
{Project Gutenberg Literary Archive Foundation}.
\newblock Project {Gutenberg}.
\newblock \url{http://www.gutenberg.org/}~(checked 2007-05-30), 2007.

\bibitem{1155030}
D.~Rotem, K.~Stockinger, and K.~Wu.
\newblock Minimizing {I/O} costs of multi-dimensional queries with bitmap
  indices.
\newblock In {\em SSDBM '06}, pages 33--44, 2006.

\bibitem{sharma2008emc}
Y.~Sharma and N.~Goyal.
\newblock An efficient multi-component indexing embedded bitmap compression for
  data reorganization.
\newblock {\em Information Technology Journal}, 7(1):160--164, 2008.

\bibitem{stockinger2002spa}
K.~Stockinger, K.~Wu, and A.~Shoshani.
\newblock Strategies for processing ad hoc queries on large data warehouses.
\newblock In {\em DOLAP'02}, pages 72--79, 2002.

\bibitem{stockinger2004esb}
K.~Stockinger, K.~Wu, and A.~Shoshani.
\newblock Evaluation strategies for bitmap indices with binning.
\newblock In {\em {DEXA} 2004}, 2004.

\bibitem{DBGEN}
{TPC}.
\newblock {DBGEN} 2.4.0.
\newblock \url{http://www.tpc.org/tpch/} (checked 2007-12-4), 2006.

\bibitem{TurneyML}
P.~D. Turney and M.~L. Littman.
\newblock Corpus-based learning of analogies and semantic relations.
\newblock {\em Machine Learning}, 60(1--3):251--278, 2005.

\bibitem{wong1985btf}
H.~K.~T. Wong, H.~F. Liu, F.~Olken, D.~Rotem, and L.~Wong.
\newblock Bit transposed files.
\newblock In {\em VLDB 85}, pages 448--457, 1985.

\bibitem{502689}
K.~Wu, E.~J. Otoo, and A.~Shoshani.
\newblock A performance comparison of bitmap indexes.
\newblock In {\em CIKM '01}, pages 559--561, 2001.

\bibitem{wu2006obi}
K.~Wu, E.~J. Otoo, and A.~Shoshani.
\newblock Optimizing bitmap indices with efficient compression.
\newblock {\em ACM Transactions on Database Systems (TODS)}, 31(1):1--38, 2006.

\end{thebibliography}

\temporary{%
\section*{todo}
\begin{itemize}
\item Decide whether to use IEC units like GiB for RAM.
\item \sout{Do we write 10,000 or 10\,000?}
\item \owen{later pass needs to check consistency in other ``dimension'' uses to
see whether ``column'' is more appropriate.}
\item \sout{Do not forget to offer a link to our source code!}
\item \sout{Owen found out that for $k>1$, column reordering can behave quite differently
than for $k=1$, presumably because of the increased density.}
\item \sout{Why does the king james file gets larger as $k$ increases? Odd, odd.
  \owen{close to explanation}}
\item\sout{Document how we adapt the values of k per dimension by using the cardinalities.}
\cut{\item 
\daniel{To limit our scope, I propose to no longer use DBLP\@. We may also want
to drop other data sets. But we shall add the king james bible data set to compensate.}}
\end{itemize}
\daniel{Daniel reports an 8 page limit, July 4 submission date.}
}
\end{document}